\documentclass[11pt,twoside]{amsart}
\usepackage{setspace}
\usepackage{amsmath,hyperref}
\usepackage{amstext}
\usepackage{amssymb}
\usepackage{amsthm} 
\usepackage{mathrsfs}
\usepackage{xcolor}
\usepackage{enumerate}   
\usepackage{tikz}
\usepackage{geometry}
\newtheorem{thm}{Theorem}[section]
\newtheorem{cor}[thm]{Corollary}
\newtheorem{lem}[thm]{Lemma}
\newtheorem{prop}[thm]{Proposition}
\theoremstyle{definition}
\newtheorem{defn}[thm]{Definition}
\newtheorem{rem}[thm]{Remark}
\newtheorem{ass}[thm]{Assumption}

\numberwithin{equation}{section}

\newcommand{\be}{\begin{equation}}
\newcommand{\ee}{\end{equation}}
\newcommand{\ba}{\begin{aligned}}
\newcommand{\ea}{\end{aligned}}
\newcommand{\cA}{\mathcal{A}}
\newcommand{\cF}{\mathcal{F}}
\newcommand{\bF}{\mathbb{F}}
\newcommand{\bP}{\mathbb{P}}
\newcommand{\bQ}{\mathbb{Q}}
\newcommand{\R}{\mathbb{R}}

\newcommand{\cS}{\mathcal{S}}
\newcommand{\cC}{\mathcal{C}}

\newcommand{\cK}{\mathcal{K}}
\newcommand{\cKs}{\mathcal{K}_{\rm s}}

\newcommand{\cU}{\mathcal{U}}

\newcommand{\cX}{\mathcal{X}}
\newcommand{\cXs}{\mathcal{X}_{\rm s}}

\newcommand{\ud}{\,\mathrm{d}}
\newcommand{\Lloc}{L^2_{\rm loc}}

\DeclareMathOperator{\tr}{Tr}

\title{A stochastic control perspective on term structure models with roll-over risk}

\author[C. Fontana]{Claudio Fontana}
\address{Department of Mathematics ``Tullio Levi - Civita'', University of Padova, Italy.}
\email{fontana@math.unipd.it}

\author[S. Pavarana]{Simone Pavarana}
\address{Department of Mathematical Stochastics, University of Freiburg, Germany.}
\email{simone.pavarana@stochastik.uni-freiburg.de}

\author[W.J. Runggaldier]{Wolfgang J. Runggaldier}
\address{Department of Mathematics ``Tullio Levi - Civita'', University of Padova, Italy.}
\email{runggal@math.unipd.it}

\date{\today}

\thanks{
The authors are grateful to Eckhard Platen, the Editor, an Associate Editor and a Reviewer for helpful comments that helped to improve the paper.
The first author gratefully acknowledges financial support from the University of Padova (research programme STARS StG PRISMA).
The second author thanks the Carl-Zeiss foundation for the financial support.}

\begin{document}

\maketitle

\begin{abstract}
In this paper, we consider a generic interest rate market in the presence of roll-over risk, which generates spreads in spot/forward term rates. 
We do not require classical absence of arbitrage and rely instead on a minimal market viability assumption, which enables us to work in the context of the benchmark approach. In a Markovian setting, we extend the control theoretic approach of \cite{GR13} and derive representations of spot/forward spreads as value functions of suitable stochastic optimal control problems, formulated under the real-world probability and with power-type objective functionals. We determine endogenously the funding-liquidity spread by relating it to the risk-sensitive optimization problem of a representative investor.
\end{abstract}

\vspace{1cm}

{\em This work is dedicated to the memory of Tomas Bj\"ork, whose scientific work and teachings had a deep impact on the Mathfinance community, in particular for what concerns the term structure of interest rates.}

\section{Introduction}		\label{sec:intro}

Over the last fifteen years, interest rate markets have been marked by two major facts: first, starting with the global financial crisis, the emergence of the ``multi-curve'' phenomenon; second, in more recent years, the reform of interest rate benchmarks.
The multi-curve phenomenon refers to the fact that interbank term rates, such as Euribor and Libor rates, exhibit different characteristics depending on their tenor (i.e., the length of the term of the underlying loan).
In particular, interbank term rates differ from risk-free forward rates by a certain spread, which depends on the specific tenor under consideration.
This is due to the presence of counterparty, funding and liquidity risks that have emerged as major sources of risk during and after the global financial crisis (see \cite{FilipovicTrolle13} and, for an overview of the topic, \cite{GrbacRunggaldier}).
The reform of interest rate benchmarks aims at overcoming the deficiencies in the mechanism determining interbank benchmark rates. Existing benchmark interest rates such as Libor rates are being gradually replaced by overnight rates, i.e., nearly risk-free rates referring to a tenor of one business day. At the current stage of the reform, one of the central issues concerns the construction and the use of term rates, i.e., rates referring to tenors that are longer than overnight. 

The multi-curve phenomenon and the question of term rates in the Libor reform are related to a common fundamental issue: the different impact of counterparty, funding and liquidity risks in lending/borrowing at term with respect to rolled-over investments at overnight frequency. We generically refer to this aspect as {\em roll-over risk}, following the recent work \cite{BMSS23}.
Apart from counterparty risk factors, that are not explicitly considered in this work, roll-over risk refers to the possibility that a borrower who needs to refinance a loan at a certain future date may only do so at an increased interest rate, due to insufficient liquidity in the money market. 
This increased interest rate is captured through a funding-liquidity spread, which represents one of the main ingredients of our framework, similarly as in \cite{BMSS23}.
In the market, the funding-liquidity spread should be priced in term rates, since they effectively hedge against roll-over risk. Therefore, even in the absence of counterparty risk, roll-over risk provides an  explanation of the multi-curve phenomenon, since the spreads associated to term rates with different tenors are due to the increased funding-liquidity risk over longer time horizons. In this sense, the present work continues on the line of \cite{FilipovicTrolle13,Gallitschke17,Alfeus20}.
The recent empirical analysis of \cite{SS22} demonstrates that roll-over risk is of equal importance to credit risk in explaining the spread between Libor rates and rates of overnight indexed swaps (OIS).

The first contribution of this work is to propose a general framework for roll-over risk in interest rate markets, considering a financial market composed by zero-coupon bonds for all maturities together with single-period swaps referencing term rates.
The study of bond markets with uncountably many assets goes back to the seminal works \cite{BDMKR97,BKR97} of T. Bj\"ork and co-authors.
In our setting, we require a minimal market viability condition, which amounts to the existence of the num\'eraire (or growth-optimal) portfolio, making use of the recent results of \cite{Kardaras22}. In particular, we develop our theory under the real-world probability, since market viability does not suffice to ensure the existence of a risk-neutral probability. 
Our framework therefore fits into the benchmark approach developed by E. Platen and co-authors (see \cite{PlatenHeath}). We show that market viability does suffice to provide general characterizations of spot and forward spreads associated to term rates.
In doing so, we extend the setup of \cite{BMSS23} by imposing less stringent no-arbitrage requirements and weaker modelling assumptions.

Our second contribution consists in showing that, in a Markovian setting, spot and forward spreads can be represented as value functions of suitable stochastic optimal control problems. This extends the results of \cite{GR13}, where a similar program has been carried out for classical (single-curve) term structures. 
As acknowledged in \cite{GR13}, the original idea of linking the term structure equation to stochastic control goes back to earlier discussions between T. Bj\"ork and the last author of this work.
In the risk-neutral setting, the approach of \cite{GR13} consists in viewing the term structure PDE as the HJB equation associated to a stochastic optimal control problem where the dynamics of the underlying Markov factor process are affected by a feedback control, thus obtaining a representation of bond prices as the corresponding value function. This method has also been extended to swap measures in \cite{Cogo13}.
Working under the real-world probability, we generalize this approach to the case of multi-curve term structures generated by roll-over risk, under the assumption that the growth-optimal portfolio has a Markovian structure. The obtained stochastic control representations are based on power-type transformations and enable us to interpret spot/forward spreads as the values of hypothetical games between a lender and a borrower (see Remark \ref{rem:interpretation} below).

Our last contribution concerns a possible approach for the determination of the funding-liquidity spread, one of the key ingredients of our setup. We start from the observation that, if a risk-free savings account is assumed to exist, as we do in our setting, then a borrowing account affected by funding-liquidity risk cannot be fairly priced by a marginal utility pricing rule based on logarithmic preferences. We then consider a more risk-averse representative investor who optimally trades according to a risk-sensitive criterion. By considering more risk-averse preferences than logarithmic ones, we assume that the representative investor prices correctly roll-over risk. This enables us to provide an equation for the funding-liquidity spread, which depends on the risk aversion of the representative investor and on the model coefficients.
Moreover, a common risk aversion coefficient can be chosen for the risk-sensitive investment problem and for the stochastic optimal control problems described above, thus providing a characterization of all quantities in the model in terms of the risk preferences of a single representative investor.

The paper is structured as follows. In Section \ref{sec:setup}, we revisit some foundational concepts underlying the benchmark approach for a generic financial market containing infinitely many assets.
We then introduce the roll-over-risk-adjusted borrowing account and spot and forward term rates, alongside with their fundamental properties. In Section \ref{sec:control}, we consider a Markovian setting and derive representations of spot and forward spreads as value functions of stochastic optimal control problems. In Section \ref{sec:sensitive}, we propose an approach to determine the funding-liquidity spread by relating it to a risk-sensitive optimization problem of a representative investor.

\section{A general interest rate market with roll-over risk}		\label{sec:setup}

In this section, we present a general setup for an interest rate market in the presence of roll-over risk. We start in Section \ref{sec:market} by discussing market viability for a generic financial market containing zero-coupon bonds for all maturities. In Section \ref{sec:ROrisk}, we introduce the roll-over-risk-adjusted borrowing account and, in Section \ref{sec:spot}, we describe its connection to term rates. Section \ref{sec:fwd} completes the description of the interest rate market by introducing forward term rates.
We let $(\Omega,\cF,\bP)$ be a probability space endowed with a right-continuous filtration $\bF=(\cF_t)_{t\geq0}$, with respect to which all processes introduced in the following are assumed to be adapted.

\subsection{Setting and market viability}	\label{sec:market}

We consider a generic financial market in an infinite time horizon where  a family $\cS=\{S^i:i\in I\}$ of assets is traded, where $I$ is a non-empty index set. We assume that all elements $S^i$, $i\in I$, are strictly positive processes with continuous paths. We furthermore assume the existence of a savings account process $S^0:=\exp(\int_0^{\cdot}r_t\ud t)$, where $r$ denotes the instantaneous risk-free rate, satisfying $\int_0^T|r_t|\ud t<+\infty$ a.s. for all $T>0$.

The family $\cS$ is assumed to include {\em zero-coupon bonds} (ZCB) for all maturities $T>0$, together with the contracts introduced in Sections \ref{sec:ROrisk} and \ref{sec:fwd} as well as possible additional securities that are not explicitly modelled in this work. 
We denote by $P(t,T)$ the price at time $t$ of a ZCB with maturity $T$, for all $0\leq t\leq T<+\infty$.
In this market setting, the index set $I$ is uncountable. At this level of generality, a complete analysis of this infinite-dimensional financial market has been developed in \cite[Section 4.3]{KK21} and \cite{Kardaras22}, on which we rely for the present subsection.

We assume that trading occurs in a self-financing way, investing in a finite but arbitrary number of securities with simple strategies. In this subsection, we use a bar notation to denote quantities discounted with respect to $S^0$. In particular, $\bar{S}^i:=S^i/S^0$ denotes the discounted price process of asset $i$, for each $i\in I$, while we denote by $\bar{V}$ the discounted value process of a generic portfolio, as defined below.
Discounted gains from trading processes are of the form
\be	\label{eq:stoch_int}
\sum_{j\in J}\int_0^{\cdot}\theta^j_t\ud \bar{S}^j_t,
\ee
where $J$ is a finite subset of $I$ and $\{\theta^j:j\in J\}$ is a collection of simple predictable processes. 
Letting $\cKs$ be the set of all processes of the form \eqref{eq:stoch_int}, the set of {\em simple portfolios} is defined as
\be	\label{eq:simple_pf}
\cXs := \bigl\{1+\bar{X} : \bar{X}\in\cKs\text{ and }1+\bar{X}_t>0\text{ a.s. for all }t\geq0\bigr\}.
\ee

In this market setup, a first and fundamental question concerns absence of arbitrage. We say that {\em market viability} holds if the following condition is satisfied:
\be	\label{eq:NUPBR}
\lim_{m\to+\infty}\sup_{\bar{V}\in\cXs}\bP(\bar{V}_T>m)=0,
\qquad\text{ for all }T>0.
\ee
Condition \eqref{eq:NUPBR} is equivalent to the notion of viability considered in \cite{Kardaras22}. 
In a finite-dimensional setup, condition  \eqref{eq:NUPBR} becomes equivalent to the {\em no unbounded profit with bounded risk} (NUPBR) condition of \cite{KaratzasKardaras07} for all finite-time horizons, restricted to simple strategies.
The fundamental theorem of asset pricing in the version of \cite[Theorem 3.3]{Kardaras22} asserts that market viability holds if and only if there exists a {\em local martingale deflator} (LMD), i.e., a strictly positive local martingale $Y$ with $Y_0=1$ such that $Y\bar{S}^i$ is a local  martingale, for every $i\in I$.
Note that the existence of an LMD implies the semimartingale property of $\bar{S}^i$, for every $i\in I$. 
In the following, without further mention, we consider market viability as a {\bf standing assumption}.

In the present market setup, it is natural to allow for the theoretical possibility of investing in infinitely many assets.\footnote{In interest rate markets containing ZCBs for all maturities, measure-valued strategies represent a possibility to define portfolios investing in infinitely many assets, as considered in \cite{BDMKR97,BKR97} also allowing for discontinuous price processes. We choose to work in the framework of \cite{KK21,Kardaras22} since it yields a workable characterization of market viability that naturally extends the theory of finite-dimensional financial markets.}
To this effect, we denote by $\cK$ the closure of $\cKs$ in the semimartingale topology (see \cite[Section III.6c]{JS03}), restricted to continuous semimartingales. Similarly to \eqref{eq:simple_pf}, the set of {\em extended portfolios} is defined as $\cX:=\{1+\bar{X} : \bar{X}\in\cK\text{ and }1+\bar{X}_t>0\text{ a.s. for all }t\geq0\}$.

\begin{thm}	\label{thm:NUPBR}
Market viability holds if and only if there exists an extended portfolio $\bar{V}^*\in\cX$ such that $1/\bar{V}^*$ is a LMD. Moreover, the process $\bar{V}/\bar{V}^*$ is a local martingale, for every $\bar{V}\in\cX$.
\end{thm}
\begin{proof}
The result follows directly from \cite[Exercise 4.48]{KK21}.
\end{proof}

In line with the literature, we call {\em num\'eraire portfolio} the portfolio $\bar{V}^*$ satisfying the properties stated in Theorem \ref{thm:NUPBR}.  In a general finite-dimensional financial market, the equivalence between condition \eqref{eq:NUPBR} and the existence of the num\'eraire portfolio has been proved in \cite{KaratzasKardaras07,KKS16}.

\begin{rem}[Market viability with extended portfolios]
While market viability is defined in \eqref{eq:NUPBR} in terms of simple portfolios, it actually holds also with respect to extended portfolios. Indeed, let $\bar{V}^*$ be the num\'eraire portfolio, which exists as long as \eqref{eq:NUPBR} holds. For any $\bar{V}\in\cX$, the process $\bar{V}/\bar{V}^*$ is a strictly positive local martingale and, by Fatou's lemma, a supermartingale. Therefore, for all $T>0$, it holds that $E[\bar{V}_T/\bar{V}^*_T]\leq1$, meaning that the set $\{\bar{V}_T/\bar{V}^*_T:\bar{V}\in\cX\}$ is bounded in $L^1$. Since boundedness in $L^1$ implies boundedness in $L^0$ and boundedness in $L^0$ is invariant by multiplication with a fixed strictly positive random variable, it follows that $\{\bar{V}_T:\bar{V}\in\cX\}$ is bounded in $L^0$ for all $T>0$, i.e., condition \eqref{eq:NUPBR} holds with respect to the set $\cX$ of extended portfolios.
\end{rem}

\begin{rem}
In our setup, the discounting asset is chosen as the savings account $S^0$ generated by the risk-free rate $r$. This is coherent with the current adoption of the {\em secured overnight financing rate} (SOFR) as the discounting rate for most cleared derivatives. Accordingly, market viability has been directly defined with respect to $S^0$-discounted quantities. 
An alternative approach, recently pursued by \cite{BalintSchweizer22}, consists in considering a discounting-invariant absence of arbitrage condition (named {\em dynamic share viability}) on the original undiscounted financial market represented by $(S^0,\cS)$, without fixing a priori the discounting unit. The corresponding theory for large financial markets containing countably many assets is developed in  \cite{BalintSchweizer20}.
\end{rem}

The num\'eraire portfolio enjoys several optimality properties. In particular, it coincides with the {\em growth-optimal portfolio} (GOP), i.e., the extended portfolio that achieves the maximal instantaneous logarithmic growth rate. This follows from \cite[Exercise 4.49]{KK21} (in a general finite-dimensional setup, the analogous property has been shown in \cite{KaratzasKardaras07}).
In Section \ref{sec:markov}, we shall consider a finite-dimensional Markovian setting and provide an explicit description of the GOP, together with an explicit characterization of the validity of condition \eqref{eq:NUPBR}.

The num\'eraire portfolio (or, equivalently, the GOP) plays a central role in the {\em benchmark approach}, see \cite[Chapter 10]{PlatenHeath}. The benchmark approach adopts the GOP as the reference (benchmark) asset and develops a pricing theory that does not rely on risk-neutral valuation. 
Indeed, it is well-known that market viability does not suffice to ensure the existence of a risk-neutral probability, while it represents the minimal requirement allowing for a meaningful solution to pricing and hedging problems as well as to optimal investment/consumption problems (see, e.g., \cite{CCFM17,CDM15,FontanaRunggaldier13,KaratzasKardaras07} in the finite-dimensional case). 
The notions of {\em fair portfolio} and {\em real-world price} are central in the benchmark approach and can be defined as in Definition \ref{def:fair}.
From now on, we shall mostly work with undiscounted quantities. We therefore denote by $V\in S^0\cX$ the undiscounted value process of a generic extended portfolio and by $V^*:=S^0\bar{V}^*$ the undiscounted value of the num\'eraire portfolio.

\begin{defn}	\label{def:fair}
A process $V\in S^0\cX$ is said to be {\em fair} if $V/V^*$ is a true martingale.
For $T>0$, if $H$ is an $\cF_T$-measurable random variable and there exists a fair process $V\in S^0\cX$ satisfying $V_T=H$ a.s., then the {\em real-world price} $\pi_t(H)$ of  $H$ at time $t$ is given by
\be	\label{eq:RW_price}
\pi_t(H) = V^*_t \, E\left[\frac{H}{V^*_T}\bigg|\cF_t\right],
\qquad\text{ for all }t\in[0,T].
\ee
\end{defn}

We shall say that a payoff $H$ is {\em fairly priced} if its market value is given by formula \eqref{eq:RW_price}.
The quantity $\pi_t(H)/V^*_t$ represents the {\em benchmarked price} of $H$ and, as a consequence of \eqref{eq:RW_price}, is a martingale. 
More generally, if $S$ denotes the price process of a generic traded asset or portfolio, the corresponding {\em benchmarked price} is given by $\hat{S}:=S/V^*$. Since $S^0/V^*=1/\bar{V}^*$ is a LMD, benchmarked prices are local martingales (and true martingales in the case of fairly priced assets).
In the following sections we shall denote by $\hat{P}(t,T):=P(t,T)/V^*_t$ the benchmarked price of a ZCB at time $t$ with maturity $T$. Since we assume that $P(\cdot,T)\in\cS$, for all $T>0$, market viability implies that benchmarked ZCB prices are local martingales.
We remark that the property that benchmarked prices are local martingales (and not only supermartingales) is specific to financial market models based on continuous asset price processes, as considered in the present paper.

\begin{rem}
[Real-world pricing of non-attainable payoffs]
\label{rem:log_price}
According to Definition \ref{def:fair}, the real-world price $\pi_t(H)$ coincides with the replication value of $H$. In the benchmark approach, the real-world pricing formula can also be extended to non-attainable claims. In this case, formula \eqref{eq:RW_price} corresponds to the marginal utility indifference price for a logarithmic utility function (see, e.g.,  \cite[Section 11.4]{PlatenHeath} and \cite[Proposition 4.7.1]{FontanaRunggaldier13}). This follows from the fact that $\bar{V}^*_T$ maximizes expected logarithmic utility at time $T$, provided that $E[\log(\bar{V}^*_T)]<+\infty$, for $T>0$.
Indeed, the first part of the proof of \cite[Proposition 4.3]{Becherer01} implies that, if $E[\log(\bar{V}^*_T)]<+\infty$, then it holds that $E[\log(\bar{V}^*_T)] \geq E[\log(\bar{V}_T)]$, for all $\bar{V}\in\mathcal{X}$ such that $\bar{V}_T>0$ a.s., with $E[\log(\bar{V}_T)]$ potentially taking the value $-\infty$.
\end{rem}

\subsection{Roll-over risk}	\label{sec:ROrisk}

As mentioned in the introduction, we consider a post-crisis financial market where default risk and funding-liquidity risk are possibly present. While default risk can be mitigated by considering fully collateralized transactions, funding-liquidity risk represents an important feature in the determination of interest rates even in the absence of counterparty risk. 

In line with \cite{BMSS23}, we model funding-liquidity risk as {\em roll-over risk}, consisting in the situation where an agent may no longer be able to access funding at the risk-free rate $r$, but only at a usually higher rate $\tilde{r}$, due for instance to insufficient liquidity in the money market. 
We denote by $\varphi$ the {\em funding-liquidity spread} process and let $\tilde{r}:=r+\varphi$. While the presence of roll-over risk corresponds to $\varphi\geq0$, we do not a priori exclude possible negative values of $\varphi$, that would correspond to situations of excess of liquidity in the money market.
Assuming that $\int_0^T|\varphi_t|\ud t<+\infty$ a.s., for all $T>0$, we define the {\em roll-over-risk-adjusted borrowing account} process by $\tilde{S}^0:=\exp(\int_0^{\cdot}\tilde{r}_t\ud t)$.

\begin{rem}[On the account process $\tilde{S}^0$]
\label{rem:funding_account}
It is important to note that $\tilde{S}^0$ does not belong to the assets $\cS$ available for trading. Indeed, the possibility of borrowing/lending at distinct risk-free rates $\tilde{r}$ and $r$ would give rise to obvious arbitrage possibilities. In line with \cite{BMSS23}, the process $\tilde{S}^0$ is only introduced as a modelling tool accounting for roll-over risk in term rates.
\end{rem}

According to \cite{BMSS23}, we define by $A(t,T)$ the market value at time $t$ of $\tilde{S}^0_T/\tilde{S}^0_t$, for $T\geq t$. In other words, $A(t,T)$ represents the value at time $t$ of the repayment at $T$ of a continuously rolled-over loan over $[t,T]$, in the presence of funding-liquidity risk. Observe that this definition implies that $A(T,T)=1$, for all $T>0$.
For every $T>0$, we assume that the process $(A(t,T))_{t\in[0,T]}$ is continuous and non-negative. For every $T>0$, we assume that the process $(A(t,T))_{t\in[0,T]}$ is continuous and nonnegative.
In \cite{BMSS23}, the quantity $A(t,T)$ is determined by risk-neutral valuation. In our context, we assume the validity of the following weaker condition.\footnote{In \cite{BMSS23}, in the absence of credit risk, $A(t,T)$ is defined as the discounted risk-neutral expectation of $\tilde{S}^0_T/\tilde{S}^0_t$, see \cite[equation (3.10)]{BMSS23}. Assumption \ref{ass:funding_account} represents a natural generalization of this definition to our context.}

\begin{ass}	\label{ass:funding_account}
For every $T>0$, the process $(A(t,T)\tilde{S}^0_t/V^*_t)_{t\in[0,T]}$ is a  local martingale.
\end{ass}

In particular, Assumption \ref{ass:funding_account} implies that the inclusion of a security with value process $A(\cdot,T)\tilde{S}^0$ in the financial market considered in Section \ref{sec:market} does not alter market viability and is consistent with the num\'eraire portfolio $V^*$.
The fact that adding a security whose price process is already a local martingale when denominated in units of $V^*$ leaves invariant the num\'eraire portfolio goes back to \cite{FilipovicPlaten09}.
By Assumption \ref{ass:funding_account}, the process $(A(t,T)\tilde{S}^0_t/V^*_t)_{t\in[0,T]}$ is a non-negative local martingale and, hence, a supermartingale by Fatou's lemma. Recalling that $A(T,T)=1$, it then follows that
\be	\label{eq:AtT}
A(t,T) \geq \frac{V^*_t}{\tilde{S}^0_t}\,E\biggl[\frac{\tilde{S}^0_T}{V^*_T}\bigg|\cF_t\biggr],
\ee
with equality holding if the roll-over-risk-adjusted borrowing account is fairly priced. The presence of roll-over risk corresponds to the situation where $A(t,T)>1$. In particular, this is the case if the process $\tilde{S}^0/V^*$ is a submartingale. As shown in the following lemma, the latter property always holds in a local sense whenever the funding-liquidity spread is non-negative.

\begin{lem}	\label{lem:ROrisk}
If $\varphi\geq0$, the process $\tilde{S}^0/V^*$ is a local submartingale.
Moreover, $\tilde{S}^0/V^*$ is a local martingale if and only if $\varphi=0$ (up to an evanescent set).
\end{lem}
\begin{proof}
The claim follows from the fact that $S^0/V^*$ is a LMD (see Theorem \ref{thm:NUPBR}) and, therefore, a local martingale, together with the fact that the process $\exp(\int_0^{\cdot}\varphi_t\ud t)$ is increasing if $\varphi\geq0$. 
\end{proof}

\subsection{Spot term rates}	\label{sec:spot}

In interest rate markets, an investor can avoid funding-liquidity risk by borrowing/lending money at a fixed term rate, instead of rolling-over the loan until the end of the term.
Therefore, roll-over risk should be implicitly taken into account in the determination of fair term rates. This perspective underlies the approach of \cite{BMSS23}, which we are going to follow in the present subsection.
For $0\leq t\leq T<+\infty$, we denote by $L(t,T)$ the {\em spot term rate}, i.e., the rate fixed at time $t$ for borrowing/lending one unit of money at $t$ with a repayment of $1+(T-t)L(t,T)$ at time $T$. 

Following \cite{BMSS23}, we determine term rates by comparing the following two possibilities:
\begin{enumerate}
\item[(i)] at time $t$ borrow one unit of money and continuously roll-over the loan until time $T$;
\item[(ii)] at time $t$ borrow one unit of money at the term rate $L(t,T)$ with repayment at time $T$.
\end{enumerate}
In equilibrium, the two alternatives (i) and (ii) should have the same present value at time $t$. Arguing as in \cite[Section 3.2]{BMSS23}, this enables us to determine the term rate $L(t,T)$. For (i), the present value at $t$ is simply given by $A(t,T)$, as explained in Section \ref{sec:ROrisk}. For (ii), since $L(t,T)$ is fixed at $t$, the present value at $t$ of the repayment at $T$ is given by $(1+(T-t)L(t,T))P(t,T)$. Therefore, by equating the two present values we obtain that
\[
L(t,T) = \frac{1}{T-t}\left(\frac{A(t,T)}{P(t,T)}-1\right).
\]

In the presence of roll-over risk, the quantity $A(t,T)$ coincides with the multiplicative {\em spot spread} $S(t,T)$ between term rates and simple forward rates, defined as follows:
\be	\label{eq:spot_spread}
S(t,T) := \frac{1+(T-t)L(t,T)}{1+(T-t)F(t,T)} = A(t,T),
\ee
where $F(t,T):=(1/P(t,T)-1)/(T-t)$ is the simple risk-free forward rate at $t$ for the time period $[t,T]$, for $0\leq t\leq T<+\infty$.
In typical market situations, we expect that $S(t,T)>1$. In view of \eqref{eq:spot_spread}, this happens if and only if $A(t,T)>1$, which is indicative of the presence of funding-liquidity risk, as explained in Section \ref{sec:ROrisk}. The spot spread can therefore be regarded as a term premium paid by the borrower in order to hedge against roll-over risk.

\begin{rem}[Relation to multi-curve models]	
As mentioned in Section \ref{sec:intro}, interest rate models where term rates $L(t,T)$ are distinct from risk-free forward rates $F(t,T)$ are called {\em multi-curve} models. In the present setting, the multi-curve structure arises intrinsically from roll-over risk, since the latter is responsible for the existence of multiplicative spot spreads.
\end{rem}

\subsection{Forward term rates}	\label{sec:fwd}

Spot term rates are not directly related to traded securities. In interest rate markets, the fundamental contract referencing term rates is a {\em single-period swap} (SPS), corresponding to the classical forward rate agreement (see, e.g., \cite[Exercise 19.1]{Bjork}). An SPS delivers the payoff $\delta(L(T,T+\delta)-K)$ at maturity $T+\delta$, where $K$ is a fixed rate and $\delta$ represents a fixed tenor.
In particular, SPSs represent the building blocks of interest rate swaps, which constitute the most important interest rate derivatives referencing term rates.

We denote by $\Pi(t;T,\delta,R)$ the market value at $t\in[0,T]$ of an SPS referencing $L(T,T+\delta)$, with fixed rate $R\in\R$ and tenor $\delta\in\Delta$, where $\Delta$ is a finite collection of tenors. 
For each $T>0$, $\delta\in\Delta$ and $R\in\R$, we assume that the process $(\Pi(t;T,\delta,R))_{t\in[0,T]}$ is continuous. Moreover, we assume that market viability holds when the financial market includes SPSs for all possible maturities, rates and tenors. This corresponds to requiring the validity of the following assumption.

\begin{ass}	\label{ass:SPS}
The process $(\Pi(t;T,\delta,R)/V^*_t)_{t\in[0,T]}$ is a local martingale, for every $T>0$, $R\in\R$ and $\delta\in\Delta$.
\end{ass}

Assumption \ref{ass:SPS} holds if $V^*$ is the num\'eraire portfolio for a financial market that includes all ZCBs and SPSs, as described in Section \ref{sec:market}. If Assumption \ref{ass:SPS} holds as a true martingale, then SPSs are fairly priced by $V^*$ (see Definition \ref{def:fair}). However, the local martingale requirement of Assumption \ref{ass:SPS} will suffice for our purposes.
We define forward term rates as follows. 

\begin{defn}	\label{def:fwd_rate}
For $0\leq t\leq T<+\infty$ and $\delta\in\Delta$, the {\em forward term rate} $L_t(T,T+\delta)$ is defined as the rate $R$ such that $\Pi(t;T,\delta,R)=0$. In particular, it holds that $L_T(T,T+\delta)=L(T,T+\delta)$.
\end{defn}

We furthermore assume that the function $R\mapsto\Pi(t;T,\delta,R)$ is affine, for all $0\leq t\leq T<+\infty$ and $\delta\in\Delta$. This assumption is standard in the literature when considering fully collateralized transactions (as in our setting) and is also consistent with the real-world pricing formula \eqref{eq:RW_price}. 
This assumption leads to the following standard representation of the market value of an SPS:
\be	\label{eq:SPS_price}
\Pi(t;T,\delta,R) = \delta\bigl(L_t(T,T+\delta)-R\bigr)P(t,T+\delta).
\ee

Market viability and Assumption \ref{ass:SPS} lead to the following property of forward term rates.\footnote{In Lemma \ref{lem:fwd_rate} we remark that the true martingale property holds if ZCBs and SPSs are fairly priced by $V^*$.}

\begin{lem}	\label{lem:fwd_rate}
Suppose that Assumption \ref{ass:SPS} holds. Then, for all $T>0$, $\delta\in\Delta$, the process 
\[
L_t(T,T+\delta)\hat{P}(t,T+\delta),
\qquad t\in[0,T],
\]
is a local martingale.
\end{lem}
\begin{proof}
As stated in Section \ref{sec:market}, market viability holds for the set of assets $\cS$ that includes all ZCBs. Hence, since $S^0/V^*$ is a LMD, the process $(P(t,T+\delta)/V^*_t)_{t\in[0,T+\delta]}$ is a local martingale. Making use of this fact, the result follows directly from Assumption \ref{ass:SPS} and formula \eqref{eq:SPS_price}.
\end{proof}

In post-crisis interest rate markets, it is often useful for modelling purposes to consider multiplicative {\em forward spreads}, rather than modelling directly forward term rates (see, e.g., \cite{CFG16}).
For $0\leq t\leq T<+\infty$ and $\delta\in\Delta$, the multiplicative forward spread $S_t(T,T+\delta)$ is defined as
\be	\label{eq:fwd_spread}
S_t(T,T+\delta) := \frac{1+\delta L_t(T,T+\delta)}{1+\delta F_t(T,T+\delta)},
\ee
where $F_t(T,T+\delta):=(P(t,T)/P(t,T+\delta)-1)/\delta$ is the simple forward rate at $t$ for the time period $[T,T+\delta]$. 
Note that, by definition, it holds that $S_T(T,T+\delta)=S(T,T+\delta)$, where the latter quantity is the multiplicative spot spread introduced in \eqref{eq:spot_spread}.

\begin{cor}	\label{cor:fwd_spread}
Suppose that Assumption \ref{ass:SPS} holds. Then, for all $T>0$, $\delta\in\Delta$, the process 
\[
S_t(T,T+\delta)\hat{P}(t,T),
\qquad t\in[0,T],
\]
is a local martingale. 
\end{cor}
\begin{proof}
Making use of definition \eqref{eq:fwd_spread} of the multiplicative forward spread, we have that
\[
S_t(T,T+\delta)\hat{P}(t,T)
= \hat{P}(t,T+\delta) + \delta L_t(T,T+\delta)\hat{P}(t,T+\delta).
\]
As a consequence of market viability, benchmarked ZCB prices are local martingales. Therefore, the local martingale property of $(S_t(T,T+\delta)\hat{P}(t,T))_{t\in[0,T]}$ follows from Lemma \ref{lem:fwd_rate}.
\end{proof}

\begin{rem}[Martingale properties under forward measures]
(1) In classical interest rate models based on the risk-neutral approach, the forward term rate $(L_t(T,T+\delta))_{t\in[0,T]}$ is a martingale under the forward measure $\bQ^{T+\delta}$ (see \cite[Lemma 23.4]{Bjork}), while  $(S_t(T,T+\delta))_{t\in[0,T]}$ is a martingale under the forward measure $\bQ^T$ (see \cite[Lemma 3.11]{CFG16}). 
Lemma \ref{lem:fwd_rate} and Corollary \ref{cor:fwd_spread} show that generalized versions of these properties hold in our context.

(2) In our setup, even if a risk-neutral probability is not assumed to exist, forward measures can still be constructed if ZCBs are fairly priced by the GOP (i.e., benchmarked ZCB prices are true martingales and not only local martingales).
In this case, for every $T>0$, the forward measure $\bQ^T$ can be defined by $\ud\bQ^T/\ud\bP = 1/(V^*_TP(0,T))$ and has the property that every value process is a $\bQ^T$-local martingale on $[0,T]$ when denominated in units of $P(\cdot,T)$.
\end{rem}

\section{Stochastic control representations in a Markovian setting}
\label{sec:control}

In this section, we specialize the general setting of Section \ref{sec:setup} to a financial market driven by a finite-dimensional Markov factor process, with the main goal of representing spot and forward spreads as value functions of stochastic optimal control problems, in the spirit of \cite{GR13}. 
In Section \ref{sec:markov}, we derive explicit dynamics for the GOP, determined by a Markov factor process. Section \ref{sec:preliminary} contains some preparatory results, including the PDEs that spot and forward spreads must satisfy in a Markovian setting as a consequence of market viability (more precisely, Assumptions \ref{ass:funding_account} and \ref{ass:SPS}). By relying on these results, stochastic control representations of spot/forward spreads are derived in Section \ref{sec:representations}, where we also discuss their economic interpretation.

\subsection{A Markovian setting}	\label{sec:markov}

Let us consider a probability space $(\Omega,\cF,\bP)$, endowed with the right-continuous filtration $\bF=(\cF_t)_{t\geq0}$, consisting of the natural filtration of a $d$-dimensional Brownian motion $W=(W_t)_{t\geq0}$ augmented with the $\bP$-null sets.
We denote by $\Lloc$ the set of $\R^d$-valued predictable processes $h=(h_t)_{t\geq0}$ such that $\int_0^T\|h_t\|^2\ud t<+\infty$ a.s., for all $T>0$.
We consider the generic financial market described in Section \ref{sec:market}, assuming that the family of assets $\cS$ includes ZCBs and SPSs for all maturities $T>0$, together with possible additional assets. 

The standing assumption of market viability in the sense of condition \eqref{eq:NUPBR} is assumed to be in force, thereby ensuring the existence of the num\'eraire portfolio (or, equivalently, the GOP). As in Sections \ref{sec:ROrisk}--\ref{sec:fwd}, we denote by $V^*$ the undiscounted value process of the GOP. By Theorem \ref{thm:NUPBR}, the process $Y:=S^0/V^*$ is a LMD.
Since $Y$ is a strictly positive local martingale with $Y_0=1$, martingale representation ensures the existence of a process $\theta\in\Lloc$ such that
\[
\ud Y_t = -Y_t\,\theta_t\ud W_t,
\qquad Y_0=1.
\] 
A straightforward application of It\^o's formula yields the following dynamics of the GOP:
\be	\label{eq:sde_GOP}
\ud V^*_t = V^*_t\bigl(r_t+\|\theta_t\|^2\bigr)\ud t + V^*_t \, \theta_t\ud W_t,
\qquad V^*_0=1.
\ee
In the following, we call {\em market price of risk} the process $\theta$. This is due to the fact that if one starts by postulating an It\^o process model for a finite family of assets and computes the GOP, then one obtains dynamics of the form \eqref{eq:sde_GOP}, where the process $\theta$ coincides with the market price of risk for the original family of assets, see e.g.  \cite[Chapter 10]{PlatenHeath} and also Section \ref{sec:RS} below.
In that setting, it can also be shown that market viability holds if and only if $\theta\in\Lloc$, with an analogous condition being true in the infinite-dimensional case (see \cite[Theorem 3.3]{Kardaras22}).

Let $X$ be a diffusion process taking values in a state space $D\subseteq\R^n$ and satisfying
\be	\label{eq:sde_X}
\ud X_t = f(t,X_t)\ud t + g(t,X_t)\ud W_t,
\qquad X_0=x_0\in D,
\ee
where the functions $f:\R_+\times D\to \R^n$ and $g:\R_+\times D\to\R^{n\times d}$ are sufficiently smooth to ensure the existence of a unique strong solution to \eqref{eq:sde_X} with the Markov property. We furthermore assume that, for each $t>0$, the distribution of $X_t$ has full support on $D$. 
We interpret the Markov process $X$ as a vector of economic factors that determine the market environment. 
In view of this interpretation, we introduce the following natural assumption.

\begin{ass}	\label{ass:Markov}
It holds that
\[
r_t = r(t,X_t),
\qquad \theta_t=\theta(t,X_t),
\qquad \varphi_t=\varphi(t,X_t),
\qquad\;\text{ for all }t\geq0,
\]
where $r:=\R_+\times D\to\R$, $\theta:\R_+\times D\to\R^d$ and $\varphi:\R_+\times D\to\R$.
\end{ass}

In the following, we shall derive several results by exploiting the fact that benchmarked price processes are local martingales (see Section \ref{sec:market}). However, to obtain a PDE characterization we will need a Markovian structure of benchmarked prices and this can only be true if the GOP itself has a Markovian structure. To this effect, we state the following lemma. For $f:\R_+\times D\to\R$ we denote respectively by $\nabla_x f$ and $Hf$ the gradient and the Hessian of $f$ with respect to $x$.

\begin{lem}	\label{lem:GOP}
Suppose that Assumption \ref{ass:Markov} holds. Let $v^*:\R_+\times D\to(0,+\infty)$ be a function of class $\cC^{1,2}$. Then $v^*(t,X_t)=V^*_t$ holds for all $t\geq0$ if and only if the function $v^*$ satisfies $v^*(0,x_0)=1$ and the following two conditions hold, for all $(t,x)\in\R_+\times D$:
\be	\label{eq:GOP_cond1}
g(t,x)^{\top}\nabla_x v^*(t,x) = v^*(t,x)\theta(t,x),
\ee
\be	\label{eq:GOP_cond2}\ba
&\partial_t v^*(t,x)+\nabla^{\top}_x v^*(t,x)\bigl(f(t,x)-g(t,x)\theta(t,x)\bigr)\\
&\quad
+\frac{1}{2}\tr\bigl(g(t,x)^{\top}Hv^*(t,x)\,g(t,x)\bigr)-v^*(t,x)r(t,x) = 0.
\ea\ee
\end{lem}
\begin{proof}
Making use of Assumption \ref{ass:Markov} and applying It\^o's formula to the function $v^*$, it can easily be seen that the process $(v^*(t,X_t))_{t\geq0}$ satisfies \eqref{eq:sde_GOP} if and only if conditions \eqref{eq:GOP_cond1}-\eqref{eq:GOP_cond2} are satisfied. The result follows from the fact that the SDE \eqref{eq:sde_GOP} admits a unique solution $V^*$.
\end{proof}

\begin{rem}	\label{rem:inf_dim}
For simplicity of presentation, we restrict our attention to a $d$-dimensional Brownian motion as the driving source of randomness. However, in financial market models containing infinitely many assets, infinite-dimensional driving processes are often considered. The present probabilistic setup can be generalized to a Wiener process $W$ taking values in a real separable Banach space (see, e.g., \cite[Chapter 4]{CarmonaTehranchi}) with no significant changes in the results of this section. 
In particular, the martingale representation theorem applicable in this case is provided by \cite[Theorem 4.1 and Remark 4.2]{CarmonaTehranchi}. 
In addition, under suitable conditions on $f$ and $g$, a Markov factor process $X$ with values in $D\subseteq\R^n$ can be defined as the unique strong solution to \eqref{eq:sde_X}, generalized to a driving Wiener process $W$ (see, e.g., \cite[Theorems 3.3 and 3.6]{GawareckiMandrekar}).
After these preliminaries, the PDE characterizations stated in Lemmata \ref{lem:GOP}, \ref{lem:PDE_spot} and \ref{lem:PDE_fwd} can be obtained analogously to the finite-dimensional case, since their proofs are essentially based on applications of It\^o's formula with respect to the finite-dimensional process $X$.
\end{rem}

\subsection{PDE characterization of spreads}	\label{sec:preliminary}

In order to represent spot/forward spreads as solutions to stochastic optimal control problems in a Markovian setting, we first need to obtain a PDE representation of spot/forward spreads. 
In the classical case of ZCB term structures considered in \cite{GR13}, the PDE correspond to the fundamental term structure equation. In our context, the PDEs for spot/forward spreads will be derived by relying on the market viability assumptions introduced in Section \ref{sec:setup}, in particular Assumptions \ref{ass:funding_account} and \ref{ass:SPS}.

In the remaining part of this section, we shall work under the following assumption, which ensures that the GOP has a Markovian structure (see Lemma \ref{lem:GOP}).

\begin{ass}	\label{ass:Markov_GOP}
There exists a function $v^*:\R_+\times D\to(0,+\infty)$ of class $\cC^{1,2}$ with $v^*(0,x_0)=1$ such that conditions \eqref{eq:GOP_cond1}-\eqref{eq:GOP_cond2} hold.
\end{ass}

We start by deriving the PDE associated to spot spreads. We recall from Section \ref{sec:spot} that $S(t,T)=A(t,T)$, for every $0\leq t\leq T<+\infty$, as shown in formula \eqref{eq:spot_spread}.

\begin{lem}	\label{lem:PDE_spot}
Suppose that Assumptions \ref{ass:funding_account}, \ref{ass:Markov} and \ref{ass:Markov_GOP} hold. 
For $T>0$, let $s^T:\R_+\times D\to\R$ be a function of class $\cC^{1,2}$. 
If
\be	\label{eq:spot_Markov}
S(t,T) = s^T(t,X_t),
\qquad\text{ for all }t\in[0,T],
\ee
then the function $s^T$ solves the following PDE, for all $(t,x)\in[0,T)\times D$:
\be	\label{eq:PDE_spot}	\ba
&\partial_t s^T(t,x)+\nabla^{\top}_xs^T(t,x)\left(f(t,x)-g(t,x)g(t,x)^{\top}\frac{\nabla_xv^*(t,x)}{v^*(t,x)}\right)	\\
&\; +\frac{1}{2}\tr\bigl(g(t,x)^{\top}Hs^T(t,x) \, g(t,x)\bigr)+\varphi(t,x)s^T(t,x)=0,
\ea	\ee
with terminal condition $s^T(T,x)=1$, for all $x\in D$.
\end{lem}
\begin{proof}
An application of the integration by parts formula implies that
\begin{align*}
\ud\frac{S(t,T)\tilde{S}^0_t}{V^*_t}
&= \frac{\tilde{S}^0_t}{V^*_t}\ud S(t,T) + S(t,T)\ud\frac{\tilde{S}^0_t}{V^*_t}
+ \ud\biggl\langle S(\cdot,T),\frac{\tilde{S}^0}{V^*}\biggr\rangle_t	\\
&=  \frac{\tilde{S}^0_t}{V^*_t}\ud S(t,T) + \frac{S(t,T)\tilde{S}^0_t}{V^*_t}(\varphi_t\ud t-\theta_t\ud W_t) - \frac{\tilde{S}^0_t}{V^*_t}\theta_t\ud\bigl\langle S(\cdot,T),W\bigr\rangle_t.
\end{align*}
By applying It\^o's formula to the function $s^T$ and making use of Assumption \ref{ass:Markov}, we then obtain
\begin{align*}
\ud\frac{S(t,T)\tilde{S}^0_t}{V^*_t}
&= \frac{\tilde{S}^0_t}{V_t^*}\bigl(\partial_t s^T(t,X_t)
+\nabla^{\top}_xs^T(t,X_t)f(t,X_t)\bigr)\ud t \\
&\quad+\frac{1}{2}\frac{\tilde{S}^0_t}{V_t^*}\tr\bigl(g(t,X_t)^{\top}Hs^T(t,X_t)\,g(t,X_t)\bigr)\ud t\\
&\quad+\frac{\tilde{S}^0_t}{V_t^*}\bigl(\varphi(t,X_t)s^T(t,X_t)-\nabla^{\top}_xs^T(t,X_t)g(t,X_t)\theta(t,X_t)\bigr)\ud t
+(\cdots)\ud W_t.
\end{align*}
As a consequence of Assumption \ref{ass:funding_account}, the process $(s^T(t,X_t)\tilde{S}^0_t/V^*_t)_{t\in[0,T]}$ is a local martingale. This implies that the finite variation term in the last equation must vanish. By Assumption \ref{ass:Markov_GOP}, the market price of risk $\theta(t,X_t)$ satisfies condition \eqref{eq:GOP_cond1}, thereby proving that $s^T$ solves PDE \eqref{eq:PDE_spot}. The terminal condition $s^T(T,x)=1$ follows from the fact that $S(T,T)=1$.
\end{proof}

In the present Markovian setup, the PDE \eqref{eq:PDE_spot} shows that the dynamics of the spot spread are dependent on the funding-liquidity spread $\varphi$. Recalling that the spot spread can be regarded as a term premium required to avoid roll-over risk, as discussed in Sections \ref{sec:ROrisk} and \ref{sec:spot}, this implies that the term premium is generated by the funding-liquidity spread.

\begin{rem}[Markovian structure under fair pricing]	
\label{rem:spot_Markov}
The Markovian structure \eqref{eq:spot_Markov} always holds if Assumptions \ref{ass:Markov} and \ref{ass:Markov_GOP} are satisfied and spot spreads are fairly priced. Indeed, in that case the process $(S(t,T)\tilde{S}^0_t/V^*_t)_{t\in[0,T]}$ is a true martingale and, therefore, $S(t,T)$ is given by $S(t,T)=v^*(t,X_t)E[\exp(\int_t^T(r(u,X_u)+\varphi(u,X_u))\ud u)/v^*(T,X_T)|\cF_t]$, corresponding to the conditional expectation appearing in the right-hand side of \eqref{eq:AtT}. By the Markov property of $X$, this expectation can always be expressed as a function of $(t,X_t)$.
However, $\mathcal{C}^{1,2}$ regularity is not guaranteed in general and needs to be checked by relying on the properties of the specific model under consideration.
\end{rem}

We now derive an analogous PDE representation of forward spreads. For the following lemma, we shall need an additional assumption on the Markovian structure of benchmarked ZCB prices.

\begin{ass}	\label{ass:Markov_ZCB}
For every $T>0$, there exists a function $\hat{p}^T:[0,T]\times D\to\R$ of class $\cC^{1,2}$ such that $P(t,T)/V^*_t=\hat{p}^T(t,X_t)$, for all $t\in[0,T]$.
\end{ass}

Similarly to Remark \ref{rem:spot_Markov}, we point out that in the present Markovian setting Assumption \ref{ass:Markov_ZCB} is always satisfied (up to the $\mathcal{C}^{1,2}$ regularity requirement) if Assumption \ref{ass:Markov_GOP} holds and ZCBs are fairly priced by the GOP, meaning that $\hat{P}(\cdot,T)=P(\cdot,T)/V^*$ is a true martingale, for every $T>0$.
Note also that, if Assumption \ref{ass:Markov_GOP} and \ref{ass:Markov_ZCB} hold, then the function $\hat{p}^T$ introduced in Assumption \ref{ass:Markov_ZCB} necessarily satisfies the condition $\hat{p}^T(T,x)=1/v^*(T,x)$.
The local martingale (and, hence, supermartingale) property of benchmarked ZCB prices also ensures that $\hat{p}^T(t,X_t)>0$ a.s., for all $t\in[0,T]$ and $T>0$.

\begin{lem}	\label{lem:PDE_fwd}
Suppose that Assumption \ref{ass:SPS}, Assumption \ref{ass:Markov_ZCB} and the assumptions of Lemma \ref{lem:PDE_spot} hold.
For $T>0$ and $\delta\in\Delta$, let $s^{T,\delta}:\R_+\times D\to\R$ be a function of class $\cC^{1,2}$. 
If
\be	\label{eq:fwd_Markov}
S_t(T,T+\delta) = s^{T,\delta}(t,X_t),
\qquad\text{ for all }t\in[0,T],
\ee
then the function $s^{T,\delta}$ solves the following PDE, for all $(t,x)\in[0,T)\times D$:
\be	\label{eq:PDE_fwd}	\ba
&\partial_t s^{T,\delta}(t,x)+\nabla^{\top}_xs^{T,\delta}(t,x)\left(f(t,x)+g(t,x)g(t,x)^{\top}\frac{\nabla_x\hat{p}^T(t,x)}{\hat{p}^T(t,x)}\right)	\\
&\; +\frac{1}{2}\tr\bigl(g(t,x)^{\top}Hs^{T,\delta}(t,x) \, g(t,x)\bigr)=0,
\ea	\ee
with terminal condition $s^{T,\delta}(T,x)=s^{T+\delta}(T,x)$, for all $x\in D$, where $s^{T+\delta}$ is as in Lemma \ref{lem:PDE_spot}.
\end{lem}
\begin{proof}
The proof is similar to that of Lemma \ref{lem:PDE_spot}. Applying integration by parts, we have that
\[
\ud \bigl(S_t(T,T+\delta)\hat{P}(t,T)\bigr)
= S_t(T,T+\delta)\ud\hat{P}(t,T) + \hat{P}(t,T)\ud S_t(T,T+\delta) + \ud\bigl\langle S_{\cdot}(T,T+\delta),\hat{P}(\cdot,T)\bigr\rangle_t.
\]
Since the first term on the right-hand side is a local martingale (recalling that benchmarked ZCB prices are local martingales), an application of It\^o's formula yields that
\begin{align*}
\ud\bigl(s^{T,\delta}(t,X_t)\hat{p}^T(t,X_t)\bigr)
&= \hat{p}^T(t,X_t) \bigl(\partial_t s^{T,\delta}(t,X_t)+\nabla^{\top}_xs^{T,\delta}(t,X_t)f(t,X_t)\bigr)\ud t\\
&\quad + \frac{1}{2}\hat{p}^T(t,X_t)\tr\bigl(g(t,X_t)^{\top}Hs^{T,\delta}(t,X_t)\,g(t,X_t)\bigr)\ud t\\
&\quad+\nabla^{\top}_xs^{T,\delta}(t,X_t)g(t,X_t)g(t,X_t)^{\top}{\nabla}_x\hat{p}^T(t,X_t)\ud t+(\cdots)\ud W_t.
\end{align*}
By Corollary \ref{cor:fwd_spread}, the process $(S_t(T,T+\delta)\hat{P}(t,T))_{t\in[0,T]}$ is a local martingale. This implies that the finite variation term in the last equation must vanish, thereby proving that the function $s^{T,\delta}$ satisfies the PDE \eqref{eq:PDE_fwd}. 
The Markovian representation \eqref{eq:spot_Markov} implies that
\[
s^{T,\delta}(T,X_T) = S_T(T,T+\delta) = S(T,T+\delta) = s^{T+\delta}(T,X_T),
\]
thus obtaining the terminal condition $s^{T,\delta}(T,x)=s^{T+\delta}(T,x)$, for all $x\in D$.
\end{proof}

\begin{rem}[Markovian structure under fair pricing]
A property analogous to Remark \ref{rem:spot_Markov} also applies to Lemma \ref{lem:PDE_fwd}. More specifically, under the validity of condition \eqref{eq:spot_Markov} and Assumption \ref{ass:Markov_ZCB}, the Markovian structure \eqref{eq:fwd_Markov} always holds if SPSs are fairly priced by the GOP (meaning that Assumption \ref{ass:SPS} holds in the stronger form of a true martingale).
Note, however, that $\mathcal{C}^{1,2}$ regularity of the function $s^{T,\delta}$ in \eqref{eq:fwd_Markov} is not ensured in general.
\end{rem}

In general, the Cauchy problems associated to the parabolic PDEs \eqref{eq:PDE_spot} and \eqref{eq:PDE_fwd} admit more than one $\cC^{1,2}$ solution. Uniqueness holds when those problems are restricted to a suitably chosen {\em uniqueness class} $\cC_{\rm un}$, namely a family of $\cC^{1,2}$ functions in which there exists at most one solution. Typical choices for the uniqueness class are the family of functions with exponentially quadratic growth or the family of non-negative functions. Within these two families of functions, under additional assumptions on the coefficients $f$ and $g$ in \eqref{eq:sde_X} and on the functions appearing in Assumption \ref{ass:Markov}, classical results ensure the existence of unique solutions to \eqref{eq:PDE_spot} and \eqref{eq:PDE_fwd} (see, e.g., Chapter 6 of \cite{Friedman75} or Chapter 6 of \cite{Pascucci}).

In \cite[Section 2.3.1]{Pavarana22}, under the assumption of linear-Gaussian dynamics for the Markov factor process $X$ and of a quadratic structure of the functions appearing in Assumption \ref{ass:Markov}, explicit solutions to the PDEs \eqref{eq:PDE_spot} and \eqref{eq:PDE_fwd} are derived. The solutions result to have an exponentially-quadratic form with coefficients determined by the solutions of suitable ODEs.

\subsection{Stochastic control representations of spot and forward spreads}
\label{sec:representations}

Making use of the results of Section \ref{sec:preliminary}, we proceed to represent spot and forward spreads as the value functions of suitable stochastic optimal control problems.

For all stochastic optimal control problems considered in the following, we use the generic notation $\cU$ to denote the set of {\em admissible controls}. More specifically, in Proposition \ref{prop:control_bond} and Theorems \ref{thm:control_spot} and \ref{thm:control_fwd}, the set $\cU$ contains all $\R^d$-valued progressively measurable processes $u$ such that the SDE defining the controlled process admits a unique weak solution and the expectation defining the objective functional is finite. We shall use the same notation $\cU$ even if the specific structure of each problem under consideration would correspond to different requirements on the controls. 

In the following, $\cC_{\rm un}$ denotes a generic uniqueness class, whose choice depends on the specific properties of the model under consideration, as explained at the end of Section \ref{sec:preliminary}. In our context, an especially relevant uniqueness class is given by the family of functions that correspond to fair prices (see Definition \ref{def:fair}). The observation that such family of functions constitutes a uniqueness class is a consequence of the fact that a true martingale is entirely determined by its terminal value. However, the results of this section are stated and valid for a generic uniqueness class $\cC_{\rm un}$.

Before considering spot/forward spreads, we first provide a stochastic control representation of benchmarked ZCB prices.
As a preliminary to the next proposition, we derive the associated PDE. We recall that benchmarked ZCB prices are local martingales. Therefore, if Assumptions \ref{ass:Markov_GOP} and \ref{ass:Markov_ZCB} hold, a straightforward application of It\^o's formula yields the following PDE:
\begin{align}
\partial_t\hat{p}^T(t,x)+\nabla^{\top}_x\hat{p}^T(t,x)f(t,x)+\frac{1}{2}\tr\bigl(g(t,x)^{\top}H\hat{p}^T(t,x)\,g(t,x)\bigr)&=0,
\quad\forall (t,x)\in[0,T)\times D,\notag\\
\hat{p}^T(T,x) &= \frac{1}{v^*(T,x)},
\quad\forall x\in D,
\label{eq:PDE_ZCB}
\end{align}
where the function $\hat{p}^T$ is as in Assumption \ref{ass:Markov_ZCB}. We can then state the following result, which can be regarded as a counterpart to \cite[Section 3.1]{GR13} under the benchmark approach.

\begin{prop}	\label{prop:control_bond}
Suppose that Assumptions  \ref{ass:Markov}, \ref{ass:Markov_GOP} and \ref{ass:Markov_ZCB} hold and there exists a unique solution $\hat{p}^T$ to \eqref{eq:PDE_ZCB} in the class $\cC_{\rm un}$.
For $T>0$, consider the following stochastic optimal control problem:
\be	\label{eq:control_bond} \left\{\ba
& \ud X^u_t=\bigl(f(t,X^u_t)+g(t,X^u_t)\,u_t\bigr)\ud t+g(t,X^u_t)\ud W_t, \\
&w^T(t,x)=\displaystyle{\min_{u\in\mathcal{U}}} \, E_{t,x}\left[\frac{1}{2}\int^T_t\|u_s\|^2\ud s+\log v^*(T,X^u_T)\right],
\ea\right.\ee
where the function $v^*$ is as in Assumption \ref{ass:Markov_GOP}.
Suppose that the value function satisfies $w^T\in\mathcal{C}^{1,2}$ and there exists an optimal control belonging to $\mathcal{U}$. Assume moreover that $\exp(-w^T)\in\mathcal{C}_{\rm un}$ and the natural candidate for the optimal control defined by
\be	\label{eq:opt_control_bond}
u^*(t,x) := -g(t,x)^{\top} \nabla_x w^T(t,x),
\qquad \text{ for all } (t,x)\in[0,T]\times D,
\ee
belongs to $\mathcal{U}$. Then it holds that $\hat{p}^T(t,x) = \exp(-w^T(t,x))$, for all $(t,x)\in[0,T]\times D$.
\end{prop}
\begin{proof}
By assumption, problem \eqref{eq:control_bond} admits an optimal control $u^{\rm opt}\in\cU$. Let us denote by $X^{\rm opt}$ the solution to the SDE in  \eqref{eq:control_bond} for $u=u^{\rm opt}$, whose existence is guaranteed by the definition of $\cU$. By the definition of the value function $w^T$, optimality of $u^{\rm opt}$ implies that the process 
\[
w^T(t,X^{\rm opt}_t) + \frac{1}{2}\int_0^t\|u^{\rm opt}_s\|^2\ud s,
\qquad t\in[0,T],
\]
is a martingale. Making use of this property, an application of It\^o's formula yields that
\be\ba
& \partial_t w^T(t,x)+\nabla^{\top}_xw^T(t,x)\bigl(f(t,x)+g(t,x)u^{\rm opt}(t,x)\bigr)+\frac{1}{2}\tr\bigl(g(t,x)^{\top}Hw^T(t,x)g(t,x)\bigr)	\\
&\qquad +\frac{1}{2}\|u^{\rm opt}(t,x)\|^2=0,
\ea\ee
for all $(t,x)\in[0,T)\times D$. 
Since by assumption the natural candidate $u^*$ for the optimal control given in \eqref{eq:opt_control_bond} belongs to $\mathcal{U}$, we can replace $u^{\rm opt}$ in the last equation by $u^*$, thus obtaining 
\be	\label{eq:control_bond_proof}	\ba
&\partial_t w^T(t,x)+\nabla^{\top}_xw^T(t,x)f(t,x)
+\frac{1}{2}\tr\bigl(g(t,x)^{\top}Hw^T(t,x)g(t,x)\bigr)	\\
&\qquad -\frac{1}{2}\nabla^{\top}_xw^T(t,x)g(t,x)g(t,x)^{\top}{\nabla}_xw^T(t,x) = 0,
\ea\ee
for all $(t,x)\in[0,T)\times D$, with terminal condition $w^T(T,x)=\log v^*(T,x)$, for all $x\in D$.
Define the function $p':[0,T]\times D\to(0,+\infty)$ by $p'(t,x):=\exp(-w^T(t,x))$, for all $(t,x)\in[0,T]\times D$. Applying this transformation to \eqref{eq:control_bond_proof}, we arrive at the following PDE:
\[
\partial_tp'(t,x)+\nabla^{\top}_xp'(t,x)f(t,x)+\frac{1}{2}\tr\bigl(g(t,x)^{\top}Hp'(t,x)\,g(t,x)\bigr)=0,
\]
for all $(t,x)\in[0,T)\times D$, with terminal condition
\[
p'(T,x) = e^{-w^T(T,x)} = \frac{1}{v^*(T,x)},
\qquad\text{ for all }x\in D.
\]
We have therefore shown that the function $p'$ solves the PDE \eqref{eq:PDE_ZCB}. 
The result of the proposition follows since by assumption $p'\in\cC_{\rm un}$ and there exists a unique solution to problem \eqref{eq:PDE_ZCB} within the family of functions $\cC_{\rm un}$.
\end{proof}

The stochastic optimal control problem considered in Proposition \ref{prop:control_bond} can be interpreted as the problem of a representative issuer of a ZCB who aims at minimizing the yield of the benchmarked ZCB (recalling that the benchmarked ZCB satisfies the terminal condition $\hat{P}(T,T)=1/V^*_T$) by changing the drift in the dynamics of the factor process through a feedback control, being thereby subject to a quadratic cost for her/his control actions. This corresponds to the analogue under the benchmark approach of the situation described in \cite[Remark 3.2]{GR13}. 

\begin{rem}	\label{rem:control_bond}
(1) In Proposition \ref{prop:control_bond}, the admissibility of the candidate optimal control \eqref{eq:opt_control_bond} is assumed. Similarly as in \cite{GR13}, it can be easily checked that the HJB equation associated to problem \eqref{eq:control_bond} coincides with the PDE \eqref{eq:PDE_ZCB} and yields a candidate optimal control of the form \eqref{eq:opt_control_bond}. Once admissibility  is verified, the result of Proposition \ref{prop:control_bond} follows. 
In our context, admissibility can be proved if suitable conditions are assumed on the coefficients of \eqref{eq:sde_X} and on the functions in Assumption \ref{ass:Markov}. However, these conditions turn out to be overly restrictive if imposed at the level of the general setup of Section \ref{sec:markov}. We therefore prefer to leave admissibility as a property that should be checked on a case by case basis, depending on the specific model under analysis. 
This remark applies also to Theorems \ref{thm:control_spot} and \ref{thm:control_fwd} below.

(3) We point out that the controlled factor process $X^u$ in problem \eqref{prop:control_bond}, and similarly in Theorems \ref{thm:control_spot} and \ref{thm:control_fwd} below, is to be considered as a purely formal process, unlike the Markov process $X$ introduced in \eqref{eq:sde_X} that can be given an economic interpretation.
\end{rem}

The stochastic optimal control problem considered in Proposition \ref{prop:control_bond} relies on a logarithmic transformation. In our context, the application of a logarithmic transformation is linked to the fact that ZCB prices are local martingales when benchmarked with respect to $V^*$, which represents the optimal portfolio for logarithmic utility (see Remark \ref{rem:log_price}). 
Aiming at deriving stochastic control representations for spot spreads, the relevant local martingale property is given by Assumption \ref{ass:funding_account}. In this case, in the presence of a non-null funding-liquidity spread $\varphi$, benchmarked spot spreads $S(t,T)/V^*_t$ are not local martingales. This suggests to consider transformations different from the logarithmic one (compare also with the discussion in Section \ref{sec:spread} below).
Inspired by ideas going back to \cite{DeFrancescoRunggaldier93}, we shall consider transformations induced by power functions, leading to a parametrized family of stochastic optimal control problems. Power transformations include both concave and convex optimization problems, depending on whether the power is smaller or greater than one. As discussed in Remark \ref{rem:interpretation} below, this will enable us to reflect the perspective of both a representative lender and a representative borrower in the determination of spot/forward spreads, thereby leading to a game theoretic interpretation of spot/forward spreads.
Moreover, a power utility function will also be adopted in Section \ref{sec:sensitive} to relate the funding-liquidity spread to the risk-sensitive problem for a representative investor.
We refer the interested reader to \cite[Section 3.1]{Pavarana22} for an application of logarithmic transformations to spot/forward spreads.

\begin{thm}	\label{thm:control_spot}
Suppose that the assumptions of Lemma \ref{lem:PDE_spot} are satisfied and there exists a unique solution $s^T$ to \eqref{eq:PDE_spot} in the class $\cC_{\rm un}$.
Let $\eta^-\in(0,1)$ and $\eta^+>1$. For $T>0$, consider the two following stochastic optimal control problems:
\begin{subequations}
\be	\label{eq:control_spot-}	\left\{\ba
& dX^{u,-}_t=\Biggl(f(t,X^{u,-}_t)-g(t,X^{u,-}_t)g(t,X^{u,-}_t)^{\top}\displaystyle{\frac{\nabla_xv^*(t,X^{u,-}_t)}{v^*(t,X^{u,-}_t)}}\Biggr.\\
&\qquad\qquad\quad\Biggl.
+\sqrt{\frac{1-\eta^-}{\eta^-}} g(t,X^{u,-}_t) u_t\Biggr)\ud t+g(t,X^{u,-}_t)\ud W_t, 	\\
& z_-^T(t,x)=\max_{u\in\mathcal{U}}\,E_{t,x}\left[\exp\left(\int^T_t \Bigl(\eta^-\varphi(s,X^{u,-}_s)-\frac{1}{2}\|u_s\|^2\Bigr)\ud s\right)\right],
\ea\right.
\ee
\be	\label{eq:control_spot+}
\left\{\ba
& dX^{u,+}_t=\Biggl(f(t,X^{u,+}_t)-g(t,X^{u,+}_t)g(t,X^{u,+}_t)^{\top}\displaystyle{\frac{\nabla_xv^*(t,X^{u,+}_t)}{v^*(t,X^{u,+}_t)}}\Biggr.\\
&\qquad\qquad\quad\Biggl.
-\sqrt{\frac{\eta^+-1}{\eta^+}} g(t,X^{u,+}_t) u_t\Biggr)\ud t+g(t,X^{u,+}_t)\ud W_t, 	\\
& z_+^T(t,x)=\min_{u\in\mathcal{U}}\,E_{t,x}\left[\exp\left(\int^T_t \Bigl(\eta^+\varphi(s,X^{u,+}_s)+\frac{1}{2}\|u_s\|^2\Bigr)\ud s\right)\right],
\ea\right.
\ee
\end{subequations}
where the function $v^*$ is as in Assumption \ref{ass:Markov_GOP}.
Suppose that $z^T_{\pm}\in\mathcal{C}^{1,2}$ and there exists optimal controls belonging to $\mathcal{U}$ for both problems \eqref{eq:control_spot-} and \eqref{eq:control_spot+}. Assume also that $(z^T_-)^{1/\eta^-}\in\mathcal{C}_{\rm un}$ and $(z^T_+)^{1/\eta^+}\in\mathcal{C}_{\rm un}$ and the natural candidates for the optimal controls defined by
\be	\label{eq:opt_control_spot}
u^{*,-}(t,x) := \sqrt{\frac{1-\eta^-}{\eta^-}}g(t,x)^{\top}\frac{\nabla_x z_-^T(t,x)}{z_-^T(t,x)},
\quad
u^{*,+}(t,x) := \sqrt{\frac{\eta^+-1}{\eta^+}}g(t,x)^{\top}\frac{\nabla_x z_+^T(t,x)}{z_+^T(t,x)}.
\ee
belong to $\mathcal{U}$. Then, for all $(t,x)\in[0,T]\times D$, it holds that
\[
\bigl(z_-^T(t,x)\bigr)^{1/\eta^-} = s^T(t,x) = \bigl(z_+^T(t,x)\bigr)^{1/\eta^+}.
\]
\end{thm}
\begin{proof}
Since problems \eqref{eq:control_spot-} and \eqref{eq:control_spot+} can be treated in full analogy, we shall only give the proof in the case of problem \eqref{eq:control_spot-}. 
By assumption, the problem admits an optimal control $u^{{\rm opt},-}\in\cU$. Let us denote by $X^{{\rm opt},-}$ the solution to the SDE in  \eqref{eq:control_spot-} for $u=u^{{\rm opt},-}$, whose existence is guaranteed by the assumption that $u^{{\rm opt},-}\in\cU$. By the definition of the value function $z^T_-$, optimality of $u^{{\rm opt},-}$ implies that the process 
\[
z_-^T(t,X^{{\rm opt},-}_t)\,e^{\int_0^t(\eta^-\varphi(s,X^{{\rm opt},-}_s)-\frac{1}{2}\|u^{{\rm opt},-}_s\|^2)\ud s},
\qquad t\in[0,T],
\]
is a martingale. 
Making use of this property, an application of It\^o's formula yields that
\begin{align*}
&\partial_t z_-^T(t,x)+\nabla^{\top}_xz_-^T(t,x)\left(f(t,x)-g(t,x)g(t,x)^{\top}\frac{\nabla_x v^*(t,x)}{v^*(t,x)}+\sqrt{\frac{1-\eta^-}{\eta^-}}g(t,x)\,u^{{\rm opt},-}(t,x)\right)	\\
&\; +\frac{1}{2}\tr\bigl(g(t,x)^{\top}Hz_-^T(t,x)\,g(t,x)\bigr)
+\eta^- z_-^T(t,x)\varphi(t,x)
-\frac{1}{2}z_-^T(t,x)\|u^{{\rm opt},-}(t,x)\|^2=0,
\end{align*}
for all $(t,x)\in[0,T)\times D$. 
Since by assumption the natural candidate $u^{*,-}$ for the optimal control given in \eqref{eq:opt_control_spot} belongs to $\mathcal{U}$, we can replace $u^{{\rm opt},-}$ in the last equation by $u^{*,-}$, thus obtaining 
\begin{align*}
&\partial_t z_-^T(t,x)+\nabla^{\top}_xz_-^T(t,x)\left(f(t,x)-g(t,x)g(t,x)^{\top}\frac{\nabla_x v^*(t,x)}{v^*(t,x)}\right)
+\frac{1}{2}\tr\bigl(g(t,x)^{\top}Hz_-^T(t,x)\,g(t,x)\bigr)	\\
&\; +\frac{1}{2}\frac{1-\eta^-}{\eta^-}\frac{\nabla^{\top}_x z^T_-(t,x)}{z^T_-(t,x)}g(t,x)g(t,x)^{\top}{\nabla}_xz_-^T(t,x) 
+ \eta^-z_-^T(t,x)\varphi(t,x) = 0,
\end{align*}
for all $(t,x)\in[0,T)\times D$, with terminal condition $z_-^T(T,x)=1$, for all $x\in D$.
Define the function $s':[0,T]\times D\to\R$ by $s'(t,x):=(z_-^T(t,x))^{1/\eta^-}$, for all $(t,x)\in[0,T]\times D$. Applying this transformation to the previous PDE, we arrive at the following PDE:
\begin{align*}
&\partial_t s'(t,x)+\nabla^{\top}_xs'(t,x)\left(f(t,x)-g(t,x)g(t,x)^{\top}\frac{\nabla_xv^*(t,x)}{v^*(t,x)}\right)	\\
&\; +\frac{1}{2}\tr\bigl(g(t,x)^{\top}Hs'(t,x) \, g(t,x)\bigr)+\varphi(t,x)s'(t,x)=0,
\end{align*}
for all $(t,x)\in[0,T)\times D$, with terminal condition $s'(T,x)=1$, for all $x\in D$.
We have therefore shown that the function $s'$ solves the PDE \eqref{eq:PDE_spot}. 
Since by assumption in the class $\cC_{\rm un}$ there exists a unique solution to \eqref{eq:PDE_spot} with terminal condition $s^T(T,x)=1$  and $(z^T_-)^{1/\eta^-}\in\cC_{\rm un}$, it follows that $(z^T_-(t,x))^{1/\eta^-}=s^T(t,x)$, for all $(t,x)\in[0,T]\times D$.
\end{proof}

Problem \eqref{eq:control_spot-} can be regarded as the problem of a representative lender with a power preference structure who aims at maximizing the value of the discounted roll-over-risk-adjusted borrowing account by affecting the dynamics of the factor process through a feedback control, being thereby subject to a quadratic cost for her/his control actions. Analogously, problem \eqref{eq:control_spot+} can be regarded as the corresponding problem of a representative borrower who aims at minimizing the discounted roll-over-risk-adjusted borrowing account.

\begin{rem}[A game theoretic interpretation of spot spreads]	
\label{rem:interpretation}
As discussed in Section \ref{sec:spot}, in the presence of funding-liquidity risk, one should have $s^T(t,x)\geq 1$. In this case, under the assumptions of Theorem \ref{thm:control_spot}, the following inequalities hold for all $\eta^-\in(0,1)$ and $\eta^+>0$:
\be	\label{eq:inequalities}
z^T_-(t,x) \leq s^T(t,x) \leq z^T_+(t,x).
\ee
We have thus obtained lower and upper bounds for the spot spread $s^T(t,x)$. Observe that these bounds can be made as tight as one wishes, since 
$
\lim_{\eta^-\to 1}z^T_-(t,x)
= \lim_{\eta^+\to 1}z^T_+(t,x)
= s^T(t,x).
$
Recalling the lender/borrower interpretation of problems \eqref{eq:control_spot-}-\eqref{eq:control_spot+} discussed above, the inequalities \eqref{eq:inequalities} suggest a game theoretic interpretation according to which the spot spread $s^T(t,x)$ represents  the value of a game between a lender and a borrower.
\end{rem}

\begin{rem}	\label{rem:eta}
The result of Theorem \ref{thm:control_spot} also holds for $\eta^+<0$. This case leads to a convex optimization problem and is fully analogous to problem \eqref{eq:control_spot+}, yielding  the same representation of spot spreads. However, the rightmost inequality in \eqref{eq:inequalities} fails to hold for $\eta^+<0$ .
\end{rem}

\begin{rem}	\label{rem:comment_drift}
In view of condition \eqref{eq:GOP_cond1}, the term $g(t,X_t^{u,\pm})^{\top}\nabla_xv^*(t,X_t^{u,\pm})/v^*(t,X_t^{u,\pm})$ present in the drift of the controlled process $X^{u,\pm}$ in problems \eqref{eq:control_spot-}-\eqref{eq:control_spot+} can be more compactly rewritten as $\theta(t,X_t^{u,\pm})$. This term arises because $S(t,T)\tilde{S}^0$ is a local martingale when denominated with respect to $V^*$, whose volatility is given by the market price of risk $\theta$.
\end{rem}

We point out that remarks analogous to Remark \ref{rem:control_bond} apply to the stochastic optimal control problems considered in Theorem \ref{thm:control_spot}, as well as to the problems considered in Theorem \ref{thm:control_fwd} below.
The next theorem provides an analogue to Theorem \ref{thm:control_spot} for forward spreads.

\begin{thm}	\label{thm:control_fwd}
Suppose that the assumptions of Lemma \ref{lem:PDE_spot} and Lemma \ref{lem:PDE_fwd} are satisfied and there exists a unique solution $s^{T,\delta}$ to \eqref{eq:PDE_fwd} in the class $\cC_{\rm un}$.
Let $\eta^-\in(0,1)$ and $\eta^+>1$.
For $T>0$ and $\delta\in\Delta$, consider the two following stochastic optimal control problems:
\begin{subequations}
\be	\label{eq:control_fwd-}	\left\{\ba
& dX^{u,-}_t=\Biggl(f(t,X^{u,-}_t)+g(t,X^{u,-}_t)g(t,X^{u,-}_t)^{\top}\displaystyle{\frac{\nabla_x\hat{p}^T(t,X^{u,-}_t)}{\hat{p}^T(t,X^{u,-}_t)}}\Biggr.\\
&\qquad\qquad\quad\Biggl.
+\sqrt{\frac{1-\eta^-}{\eta^-}} g(t,X^{u,-}_t) \, u_t\Biggr)\ud t+g(t,X^{u,-}_t)\ud W_t, 	\\
& z_-^{T,\delta}(t,x)=\max_{u\in\mathcal{U}}\,E_{t,x}\left[\exp\left(\eta^-\log s^{T+\delta}(T,X^{u,-}_T)-\frac{1}{2}\int^T_t\|u_s\|^2\ud s\right)\right],
\ea\right.\ee
\be	\label{eq:control_fwd+}
\left\{\ba
& dX^{u,+}_t=\Biggl(f(t,X^{u,+}_t)+g(t,X^{u,+}_t)g(t,X^{u,+}_t)^{\top}\displaystyle{\frac{\nabla_x\hat{p}^T(t,X^{u,+}_t)}{\hat{p}^T(t,X^{u,+}_t)}}\Biggr.\\
&\qquad\qquad\quad\Biggl.
-\sqrt{\frac{\eta^+-1}{\eta^+}} g(t,X^{u,+}_t) \, u_t\Biggr)\ud t+g(t,X^{u,+}_t)\ud W_t,	\\
& z_+^{T,\delta}(t,x)=\min_{u\in\mathcal{U}}\,E_{t,x}\left[\exp\left(\eta^+\log s^{T+\delta}(T,X^{u,+}_T)+\frac{1}{2}\int^T_t\|u_s\|^2\ud s\right)\right],
\ea\right.\ee
\end{subequations}
where the function $\hat{p}^T$ is as in Assumption \ref{ass:Markov_ZCB}.
Suppose that $z^{T,\delta}_{\pm}\in\mathcal{C}^{1,2}$ and there exists optimal controls belonging to $\mathcal{U}$ for both problems \eqref{eq:control_fwd-} and \eqref{eq:control_fwd+}. Assume also that $(z^{T,\delta}_-)^{1/\eta^-}\in\mathcal{C}_{\rm un}$ and $(z^{T,\delta}_+)^{1/\eta^+}\in\mathcal{C}_{\rm un}$ and the natural candidates for the optimal controls defined by
\be	\label{eq:opt_control_fwd}
u^{*,-}(t,x) := \sqrt{\frac{1-\eta^-}{\eta^-}}g(t,x)^{\top}\frac{\nabla_x z_-^{T,\delta}(t,x)}{z_-^{T,\delta}(t,x)},
\quad
u^{*,+}(t,x) := \sqrt{\frac{\eta^+-1}{\eta^+}}g(t,x)^{\top}\frac{\nabla_x z_+^{T,\delta}(t,x)}{z_+^{T,\delta}(t,x)}
\ee
belong to $\mathcal{U}$. Then it holds that 
\[
\bigl(z_-^{T,\delta}(t,x)\bigr)^{1/\eta^-} = s^{T,\delta}(t,x) = \bigl(z_+^{T,\delta}(t,x)\bigr)^{1/\eta^+}.
\]
\end{thm}
\begin{proof}
The proof is similar to that of Theorem \ref{thm:control_spot}. Since problems \eqref{eq:control_fwd-} and \eqref{eq:control_fwd+} have the same structure, we only consider here the case $\eta^-$. By assumption, problem \eqref{eq:control_fwd-} admits an optimal control $u^{{\rm opt},-}\in\mathcal{U}$. We denote by $X^{{\rm opt},-}$ the solution to the SDE in \eqref{eq:control_fwd-} for $u=u^{{\rm opt},-}$. By definition of $z^{T,\delta}_-$, optimality of $u^{{\rm opt},-}$ implies that the process 
\[
z_-^{T,\delta}(t,X^{{\rm opt},-}_t)e^{-\frac{1}{2}\int_0^t\|u^{{\rm opt},-}_s\|^2\ud s},
\qquad t\in[0,T],
\]
is a martingale. 
Making use of this property, an application of It\^o's formula yields that
\begin{align*}
&\partial_t z_-^{T,\delta}(t,x)+\nabla^{\top}_xz_-^{T,\delta}(t,x)\left(f(t,x)+g(t,x)g(t,x)^{\top}\frac{\nabla_x\hat{p}^T(t,x)}{\hat{p}^T(t,x)}+\sqrt{\frac{1-\eta^-}{\eta^-}}g(t,x)u^{{\rm opt},-}(t,x)\right)	\\
&\; +\frac{1}{2}\tr\bigl(g(t,x)^{\top}Hz_-^{T,\delta}(t,x)\,g(t,x)\bigr)
-\frac{1}{2}z_-^{T,\delta}(t,x)\|u^{{\rm opt},-}(t,x)\|^2=0,
\end{align*}
for all $(t,x)\in[0,T)\times D$. 
Since by assumption the natural candidate $u^{*,-}$ for the optimal control given in \eqref{eq:opt_control_fwd} belongs to $\mathcal{U}$, we can replace $u^{{\rm opt},-}$ in the last equation by $u^{*,-}$, thus obtaining
\begin{align*}
&\partial_t z_-^{T,\delta}(t,x)+\nabla^{\top}_xz_-^{T,\delta}(t,x)\left(f(t,x)+g(t,x)g(t,x)^{\top}\frac{\nabla_x \hat{p}^T(t,x)}{\hat{p}^T(t,x)}\right)
+\frac{1}{2}\tr\bigl(g(t,x)^{\top}Hz_-^{T,\delta}(t,x)\,g(t,x)\bigr)	\\
&\; +\frac{1}{2}\frac{1-\eta^-}{\eta^-}\frac{\nabla^{\top}_x z^{T,\delta}_-(t,x)}{z^{T,\delta}_-(t,x)}g(t,x)g(t,x)^{\top}{\nabla}_xz_-^{T,\delta}(t,x) 
 = 0,
\end{align*}
for all $(t,x)\in[0,T)\times D$, with terminal condition $z_-^{T,\delta}(T,x)=(s^{T+\delta}(T,x))^{\eta^-}$, for all $x\in D$, where the function $s^{T+\delta}$ is given in Lemma \ref{lem:PDE_spot}.
Define the function $s':[0,T]\times D\to\R$ by $s'(t,x):=(z_-^{T,\delta}(t,x))^{1/\eta^-}$, for all $(t,x)\in[0,T]\times D$. Applying this transformation to the previous PDE, we arrive at the following PDE:
\[
\partial_t s'(t,x)+\nabla^{\top}_xs'(t,x)\left(f(t,x)+g(t,x)g(t,x)^{\top}\frac{\nabla_x\hat{p}^T(t,x)}{\hat{p}^T(t,x)}\right)	
+\frac{1}{2}\tr\bigl(g(t,x)^{\top}Hs'(t,x) \, g(t,x)\bigr)=0,
\]
for all $(t,x)\in[0,T)\times D$, with terminal condition $s'(T,x)=s^{T+\delta}(T,x)$, for all $x\in D$.
We have therefore shown that the function $s'$ solves the PDE \eqref{eq:PDE_fwd}.
Since by assumption in the class $\cC_{\rm un}$ there exists a unique solution to \eqref{eq:PDE_fwd} with terminal condition $s^{T,\delta}(T,x)=s^{T+\delta}(T,x)$ and $(z^{T,\delta}_-)^{1/\eta^-}\in\cC_{\rm un}$, it follows that $z^{T,\delta}_-(t,x)=(s^{T,\delta}(t,x))^{\eta^-}$, for all $(t,x)\in[0,T]\times D$.
\end{proof}

Problems \eqref{eq:control_fwd-}-\eqref{eq:control_fwd+} can be interpreted similarly to problems \eqref{eq:control_spot-}-\eqref{eq:control_spot+}, in line with the interpretation presented after Theorem \ref{thm:control_spot}. The main difference is that, while problems \eqref{eq:control_spot-}-\eqref{eq:control_spot+} refer to {\em rolled-over} borrowing/lending operations, problems  \eqref{eq:control_fwd-}-\eqref{eq:control_fwd+} refer to {\em term} borrowing/lending.
In this perspective, problem  \eqref{eq:control_fwd-} (problem \eqref{eq:control_fwd+}, resp.) reflects the viewpoint of a representative lender (borrower, resp.) who aims at maximizing (minimizing, resp.) the term premium for tenor $\delta$, being subject to a quadratic penalization.
A remark analogous to Remark \ref{rem:interpretation} (as well as Remark \ref{rem:eta}) holds in the case of Theorem \ref{thm:control_fwd}, leading to a possible game theoretic interpretation of forward spreads.

\begin{rem}
The term $g(t,X^{u,\pm}_t)^{\top}\nabla_x\hat{p}^T(t,X_t^{u,\pm})/\hat{p}^T(t,X_t^{u,\pm})$ in the drift of the controlled process $X^{u,\pm}$ in problems \eqref{eq:control_fwd-}-\eqref{eq:control_fwd+} corresponds to the volatility of a benchmarked ZCB with maturity $T$. This term arises because the forward spread $S_{\cdot}(T,T+\delta)$ is a local martingale when multiplied by $\hat{P}(\cdot,T)$, see Corollary \ref{cor:fwd_spread}.
\end{rem}

\section{Risk-sensitive preferences and the funding-liquidity spread}	\label{sec:sensitive}

In the previous sections, we have described a generic interest rate market where ZCBs and SPSs referencing spot term rates are traded, under the weak assumption of market viability. We have seen in Sections \ref{sec:ROrisk}-\ref{sec:spot} that the presence of roll-over risk, encoded in the funding-liquidity spread $\varphi$, allows for an explanation of the multi-curve phenomenon, i.e., the existence of (spot and forward) spreads between (spot and forward) term rates and simple risk-free forward rates. However, the funding-liquidity spread $\varphi$ has been considered until now as an exogenous quantity (an exogenously given function in the Markovian setting of Section \ref{sec:control}). 

In this section, we aim at providing a possible way to determine the funding-liquidity spread by referring to the preference structure of a representative investor with risk-sensitive preferences.
This approach is related to the stochastic control representations of spot and forward spreads obtained in Theorems \ref{thm:control_spot} and \ref{thm:control_fwd}, which provide a link between spot/forward spreads and power-type preferences of a representative lender/borrower. 
As in Theorems \ref{thm:control_spot} and \ref{thm:control_fwd}, our representative investor will have power-type preferences.

To carry out this program, we first consider in Section \ref{sec:RS} the risk-sensitive optimal investment problem of a representative investor, in the context of a finite-dimensional Markovian model of a financial market, relying on the results of \cite{Nagai03}. By relying on these preparatory results, we shall show in Section \ref{sec:spread} how the funding-liquidity spread can be endogenously determined and, in particular, related to the risk aversion coefficient of the representative investor.

\subsection{A risk-sensitive optimal investment problem}	\label{sec:RS}

We consider the Markovian setting introduced in Section \ref{sec:markov}, with a factor process $X$ satisfying \eqref{eq:sde_X}. 
The risk-free savings account is given by $S^0=\exp(\int_0^{\cdot}r_t\ud t)$, where $r_t=r(t,X_t)$ as in Assumption \ref{ass:Markov}.
We restrict our attention to a finite time horizon $T$.
We specify further the financial market by assuming that the family of assets $\cS$ consists in a finite set of $m$ assets, for instance composed of ZCBs and SPSs for a finite number of maturities.
For each $i=1,\ldots,m$, the price process $S^i$ of the $i$-th asset is given by
\be	\label{eq:sde_S}
\ud S^i_t = S^i_t \, \mu^i(t,X_t)\ud t + S^i_t \, \sigma^i(t,X_t)\ud W_t,
\qquad S^i_0=s^i_0>0,
\ee
where the functions $\mu^i:[0,T]\times\R_+\to\R$ and $\sigma^i:[0,T]\times\R_+\to\R^d$ are sufficiently smooth to ensure the existence of a unique strong solution to \eqref{eq:sde_S}, for each $i=1,\ldots,m$.
To exclude the possibility of redundant assets, we assume that the matrix $\sigma(t,X_t)\in \R^{m\times d}$ is of full rank a.s. for all $t\in[0,T]$. We do not require the financial market to be complete.
Note that there is no loss of generality in assuming that the SDEs \eqref{eq:sde_X} and \eqref{eq:sde_S} are driven by the same $d$-dimensional Brownian motion $W$. 

If the asset prices dynamics are given by \eqref{eq:sde_S}, the market price of risk $\theta$ has the Markovian structure stated in Assumption \ref{ass:Markov}. More specifically, it holds that
\[
\theta_t = \sigma^+(t,X_t)\bigl(\mu(t,X_t)-r(t,X_t)\mathbf{1}\bigr)
=: \theta(t,X_t),
\qquad\text{ for all }t\in[0,T],
\]
where $\sigma^+(t,X_t)$ is the Moore-Penrose pseudoinverse of the matrix $\sigma(t,X_t)$, $\mathbf{1}=(1,\ldots,1)^{\top}\in\R^m$ and $\mu(t,X_t)=(\mu^1(t,X_t),\ldots,\mu^m(t,X_t))^{\top}$.
As already mentioned in Section \ref{sec:markov}, in the present finite-dimensional financial market the market viability requirement \eqref{eq:NUPBR} is equivalent to the validity of the condition $\theta\in\Lloc$, i.e., $\int_0^T\|\theta_t\|^2\ud t<+\infty$ a.s. (see, e.g., \cite[Corollary 4.3.19]{FontanaRunggaldier13}). This condition is assumed to be in force until the end of this section.

We suppose that the representative investor can trade in the $m$ available assets by means of self-financing strategies described by $\R^m$-valued predictable processes $\pi=(\pi_t)_{t\in[0,T]}$, where $\pi^i_t$ represents the proportion of wealth invested in the $i$-th risky asset, for each $i=1,\ldots,m$. 
The wealth process $V^{\pi}=(V^{\pi}_t)_{t\in[0,T]}$ associated to a self-financing strategy $\pi$ satisfies
\be	\label{eq:sde_V}
\frac{\ud V^{\pi}_t}{V^{\pi}_t}
= r(t,X_t)\ud t + \pi^{\top}_t\bigl(\mu(t,X_t)-r(t,X_t)\mathbf{1}\bigr)\ud t + \pi^{\top}_t\sigma(t,X_t)\ud W_t,
\ee
with initial wealth conventionally set at $V^{\pi}_0=1$.
The set $\cA$ of admissible strategies is defined as
\[
\cA := \left\{ \pi=(\pi_t)_{t\in[0,T]}\text{ $\R^m$-valued and predictable such that }\int_0^T\|\pi^{\top}_t\sigma(t,X_t)\|^2<+\infty\text{ a.s.} \right\}.
\]
Note that, as a consequence of Cauchy-Schwarz's inequality, if market viability holds then equation \eqref{eq:sde_V} is well-posed for every $\pi\in\cA$ (see, e.g., \cite[Lemma 4.3.21]{FontanaRunggaldier13}).

In this financial market, we consider a representative investor aiming at solving the following optimal investment problem:
\be	\label{eq:RS_problem}
E\bigl[(V^{\pi}_T)^{\gamma}\bigr] = \min!
\qquad\text{ over all }\pi\in\cA \text{ such that }E\bigl[(V^{\pi}_T)^{\gamma}\bigr]<+\infty.
\ee
The parameter $\gamma<0$ represents a risk aversion parameter. Problem \eqref{eq:RS_problem} belongs to the class of risk-sensitive investment problems (we refer to \cite{DavisLleo14} for a complete overview on the topic).

In the present Markovian setting, a complete characterization of the solution to problem \eqref{eq:RS_problem} has been derived in \cite{Nagai03}, under the assumption that the functions appearing in \eqref{eq:sde_X} and \eqref{eq:sde_S} are globally Lipschitz, together with additional technical assumptions that we implicitly assume here to be verified (see condition (2.4) and the assumptions of Theorem 2.1 in \cite{Nagai03}).
\cite[Proposition 2.1]{Nagai03} shows that the optimal strategy $\pi^*_t$ has the following structure:
\begin{align}
\pi^{\gamma}_t = \pi(t,X_t)
&= \frac{1}{1-\gamma}\bigl(\sigma(t,X_t)\sigma(t,X_t)^{\top}\bigr)^{-1}\sigma(t,X_t)\bigl(\theta(t,X_t)+\gamma g(t,X_t)^{\top}\Xi(t,X_t)\bigr)	
\notag\\
&= \frac{\pi^*_t}{1-\gamma} + \frac{\gamma}{1-\gamma}\bigl(\sigma(t,X_t)\sigma(t,X_t)^{\top}\bigr)^{-1}\sigma(t,X_t)g(t,X_t)^{\top}\Xi(t,X_t),
\label{eq:strategy_RS}
\end{align}
where $\pi^*_t$ is the growth-optimal strategy (generating the portfolio process $V^*$, see \cite[Chapter 10]{PlatenHeath}) and $\Xi(t,X_t)$ is the gradient of the solution to the Bellman equation (2.14) in \cite{Nagai03} with respect to the components of the factor process $X$.

For later use, we compute the dynamics of the optimal portfolio process $V^{\pi^{\gamma}}$ associated to the optimal risk-sensitive strategy $\pi^{\gamma}$ given in \eqref{eq:strategy_RS}:
\be	\label{eq:pf_dynamics}	\ba
\frac{\ud V^{\pi}_t}{V^{\pi}_t}
&= r(t,X_t)\ud t + \frac{\theta(t,X_t)^{\top}}{1-\gamma}\left(\theta(t,X_t)+\gamma g(t,X_t)^{\top}\Xi(t,X_t)\right)\ud t \\
&\quad + \frac{1}{1-\gamma}\left(\theta(t,X_t)+\gamma\sigma^+(t,X_t)\sigma(t,X_t)g(t,X_t)^{\top}\Xi(t,X_t)\right)\ud W_t.
\ea	\ee

\begin{rem}[Structure of the optimal strategy]
\label{rem:RS_strategy}
As can be seen from formula \eqref{eq:strategy_RS}, the optimal strategy $\pi^{\gamma}_t$ is invested into two mutual funds: the GOP and an additional portfolio that represents an intertemporal hedging component, which arises due to the randomness generated by the factor process $X$. The proportion of wealth allocated to the two funds varies according to the risk aversion coefficient $\gamma$.
In the limit for $\gamma\to0$, the optimal strategy $\pi^{\gamma}_t$ reduces to the growth-optimal strategy $\pi^*_t$, i.e., the optimal strategy for logarithmic preferences.
\end{rem}

\begin{rem}[On the risk aversion parameter]
(1) Problem \eqref{eq:RS_problem} can also be analyzed for $\gamma\in(0,1)$. However, in order to rely on the results of \cite{Nagai03}, we need to restrict our attention to the case $\gamma<0$. This is also justified by the fact that, for $\gamma\in(0,1)$, the risk-sensitive criterion \eqref{eq:RS_problem} is more prone to risk in comparison to logarithmic preferences, while for the purposes of Section \ref{sec:spread} we are interested in preference structures that exhibit a greater degree of risk aversion in comparison to logarithmic preferences.

(2) As mentioned in Remark \ref{rem:eta}, the control problems considered in Theorems \ref{thm:control_spot} and \ref{thm:control_fwd} can also be solved for $\eta^+<0$, since this choice leads to a convex optimization problem with the same structure and solution as in the case $\eta^+>1$. Choosing $\eta^+<0$ enables us to work with a common risk aversion parameter $\eta^+=\gamma$. This choice corresponds to considering a single representative agent who is facing two possible control setups: one that concerns the minimization of the funding-liquidity spread (problem \eqref{eq:control_spot+}) or the term premium (problem \eqref{eq:control_fwd+}), and one that concerns a risk-sensitive optimal investment as considered in \eqref{eq:RS_problem}. In both cases, the parameter $\eta^+=\gamma$ encodes the risk attitude of the representative investor.
\end{rem}

\subsection{The funding-liquidity spread}	\label{sec:spread}

We now rely on the solution of the representative investor's risk-sensitive problem \eqref{eq:RS_problem} to determine the funding-liquidity spread $\varphi$.

We start by observing that, 
in the presence of roll-over risk (i.e., if $\varphi>0$), the process $\tilde{S}^0/V^*$ fails to be a local martingale, as shown in Lemma \ref{lem:ROrisk}. As discussed in Remark \ref{rem:log_price}, using $1/V^*$ as LMD can be regarded as adopting a marginal utility pricing rule based on a logarithmic utility function $U(x)=\log(x)$. Indeed, the num\'eraire portfolio $V^*$ coincides with the log-optimal portfolio (when the latter is well-defined, see Remark \ref{rem:log_price}) and $U'(V^*_t)=1/V^*_t$. According to this viewpoint, the fact that $U'(V^*)\tilde{S}^0$ fails to be a local martingale can be interpreted as an evidence of the fact that the roll-over-risk-adjusted borrowing account $\tilde{S}^0$ is not priced correctly by a representative investor with logarithmic preferences. 
One could say that logarithmic preferences exhibit a ``myopia'' that does not allow to ``see'' properly roll-over risk and, therefore, do not justify the existence of a funding-liquidity spread $\varphi$.

On the basis of this reasoning, we modify the preference structure of our representative investor, replacing logarithmic preferences with power-type preferences, as considered in \eqref{eq:RS_problem}. The underlying idea is that a risk-sensitive representative investor, being more risk averse than a logarithmic investor, should correctly price funding-liquidity risk.
This leads us to consider a marginal pricing rule associated to a utility function of the form $U(x)=x^{\gamma}$, with $\gamma<0$, corresponding to the risk-sensitive preferences considered in Section \ref{sec:RS}. This preference structure yields the marginal utility process
\be	\label{eq:V_power}
Y_t := \bigl(V_t^{\pi^{\gamma}}\bigr)^{\gamma-1},
\qquad\text{ for }t\in[0,T].
\ee
As shown in the next proposition, the assumption that $Y$ correctly prices the roll-over-risk-adjusted borrowing account $\tilde{S}^0$, in the sense that the product $Y\tilde{S}^0$ is a local martingale, leads to an explicit expression for the funding-liquidity spread $\varphi$ in terms of the risk aversion parameter $\gamma$ and of the solution to the risk-sensitive optimal investment problem.
For brevity of notation, we use the shorter notation $\varphi_t$ to denote $\varphi(t,X_t)$, and similarly for all other processes.

\begin{prop}	\label{prop:RS_spread}
Let $Y$ be defined as in \eqref{eq:V_power}, where  $V^{\pi^{\gamma}}$ is the optimal portfolio process for the risk-sensitive problem \eqref{eq:RS_problem}. Then $Y\tilde{S}^0$ is a local martingale if and only if the funding-liquidity spread $\varphi$ satisfies
\be	\label{eq:RS_spread}
\varphi_t = -\gamma r_t + \theta^{\top}_t\bigl(\theta_t+\gamma  g_t^{\top}\Xi_t\bigr) - \frac{2-\gamma}{2(1-\gamma)}\|\theta_t+\gamma g_t^{\top}\Xi_t\|^2,
\qquad\text{ for all }t\in[0,T].
\ee
\end{prop}
\begin{proof}
As a first step, we compute the dynamics of the process $Y$ by means of It\^o's formula:
\[
\frac{\ud Y_t}{Y_t} = (\gamma-1)\left(r_t + \frac{\theta^{\top}_t}{1-\gamma}\bigl(\theta_t+\gamma g^{\top}_t\Xi_t\bigr) +\frac{1}{2}\frac{\gamma-2}{1-\gamma}\|\theta_t+\gamma g^{\top}_t\Xi_t\|^2\right)\ud t
+ \frac{1}{1-\gamma}\bigl(\theta_t+\gamma g^{\top}_t\Xi_t\bigr)\ud W_t.
\]
Then, applying integration by parts, we obtain that $Y\tilde{S}^0$ is a local martingale if and only if
\[
r_t + \varphi_t = (1-\gamma)r_t + \theta^{\top}_t(\theta_t+\gamma  g_t^{\top}\Xi_t) - \frac{1}{2}\frac{2-\gamma}{1-\gamma}\|\theta+\gamma g^{\top}_t\Xi_t\|^2,
\]
which is equivalent to condition \eqref{eq:RS_spread}.
\end{proof}

We can observe that for $\gamma=0$ (corresponding to the limiting case of logarithmic preferences, see Remark \ref{rem:RS_strategy}), the  funding-liquidity spread resulting from equation \eqref{eq:RS_spread} is null, confirming the interpretation discussed at the beginning of this subsection.

\section{Conclusions}

In this work, we have proposed a general perspective on interest rate markets affected by roll-over risk. After providing a general view on term structure models with roll-over risk, in the philosophy of the benchmark approach, we have focused on representing spot/forward spreads as value functions of suitable stochastic optimal control problems. The stochastic control formulation enables us to view the values of the spreads as the result of optimization problems of representative lenders/borrowers. A key quantity in our approach is represented by the funding-liquidity spread. We have proposed a way to endogenously determine the latter quantity by relating it to the risk aversion of a representative investor with risk-sensitive preferences.

Among the possible further developments of our work, we believe that it would be interesting to formulate an equilibrium model where market participants are subject to roll-over risk, for instance along the lines of \cite{GarleanuPedersen11}. This equilibrium framework would be more elaborate than the simple approach outlined in Section \ref{sec:sensitive} and would provide an endogenous explanation for the appearance of a funding-liquidity spread in interest rate markets.

\bibliographystyle{alpha-abbrvsort}
\bibliography{bib_FPR}

\begin{thebibliography}{CCFM17}

\bibitem[AGS20]{Alfeus20}
M.~Alfeus, M.~Grasselli, and E.~Schl\"ogl.
\newblock A consistent stochastic model of the term structure of interest rates
  for multiple tenors.
\newblock {\em Journal of Economic Dynamics and Control}, 114:103861, 2020.

\bibitem[BMSS23]{BMSS23}
A.~Backwell, A.~Macrina, E.~Schl\"ogl, and D.~Skovmand.
\newblock Term rates, multicurve term structures and overnight rate benchmarks:
  a roll-over risk approach.
\newblock {\em Frontiers of Mathematical Finance}, 2(3):340--384, 2023.

\bibitem[BS20]{BalintSchweizer20}
D.A. Balint and M.~Schweizer.
\newblock Large financial markets, discounting, and no asymptotic arbitrage.
\newblock {\em Theory of Probability and its Applications}, 65:191--223, 2020.

\bibitem[BS22]{BalintSchweizer22}
D.A. Balint and M.~Schweizer.
\newblock Making no-arbitrage discounting-invariant: a new {FTAP} version
  beyond {NFLVR} and {NUPBR}.
\newblock {\em Frontiers of Mathematical Finance}, 1(2):249--286, 2022.

\bibitem[Bec01]{Becherer01}
D.~Becherer.
\newblock The numeraire portfolio for unbounded semimartingales.
\newblock {\em Finance and Stochastics}, 5(3):327--341, 2001.

\bibitem[Bj{\"o}20]{Bjork}
T.~Bj{\"o}rk.
\newblock {\em Arbitrage Theory in Continuous Time}.
\newblock Oxford University Press, Oxford, 4th edition, 2020.

\bibitem[BDKR97]{BDMKR97}
T.~Bj\"ork, G.~{Di Masi}, Y.~Kabanov, and W.J. Runggaldier.
\newblock Towards a general theory of bond markets.
\newblock {\em Finance and Stochastics}, 1:141--174, 1997.

\bibitem[BKR97]{BKR97}
T.~Bj\"ork, Y.~Kabanov, and W.J. Runggaldier.
\newblock Bond market structure in the presence of marked point processes.
\newblock {\em Mathematical Finance}, 7(2):211--223, 1997.

\bibitem[CT06]{CarmonaTehranchi}
R.~Carmona and M.~Tehranchi.
\newblock {\em Interest Rate Models: An Infinite Dimensional Stochastic
  Analysis Perspective}.
\newblock Springer, Berlin - Heidelberg, 2006.

\bibitem[CCFM17]{CCFM17}
H.N. Chau, A.~Cosso, C.~Fontana, and O.~Mostovyi.
\newblock Optimal investment with intermediate consumption under no unbounded
  profit with bounded risk.
\newblock {\em Journal of Applied Probability}, 54(3):710--719, 2017.

\bibitem[CDM15]{CDM15}
T.~Choulli, J.~Deng, and J.~Ma.
\newblock How non-arbitrage, viability and num\'eraire portfolio are related.
\newblock {\em Finance and Stochastics}, 19:719--741, 2015.

\bibitem[CGR13]{Cogo13}
R.~Cogo, A.~Gombani, and W.J. Runggaldier.
\newblock Stochastic control and pricing under swap measures.
\newblock In R.C. Dalang, M.~Dozzi, and F.~Russo, editors, {\em Seminar on
  Stochastic Analysis, Random Fields and Applications {VII}}, pages 363--380.
  Birkh\"auser, Basel, 2013.

\bibitem[CFG16]{CFG16}
C.~Cuchiero, C.~Fontana, and A.~Gnoatto.
\newblock A general {HJM} framework for multiple yield curve modeling.
\newblock {\em Finance and Stochastics}, 20(2):267--320, 2016.

\bibitem[DL15]{DavisLleo14}
M.H.A. Davis and S.~Lleo.
\newblock {\em Risk-Sensitive Investment Management}.
\newblock World Scientific, Singapore, 2015.

\bibitem[DR93]{DeFrancescoRunggaldier93}
C.~{De Francesco} and W.J. Runggaldier.
\newblock On logarithmic and other transformations in stochastic control.
\newblock In {\em Proceedings of the 29th annual conference of the Operational
  Research Society of New Zealand}, 1993.

\bibitem[FP09]{FilipovicPlaten09}
D.~Filipovi\'c and E.~Platen.
\newblock Consistent market extensions under the benchmark approach.
\newblock {\em Mathematical Finance}, 19(1):41--52, 2009.

\bibitem[FT13]{FilipovicTrolle13}
D.~Filipovi\'c and A.B. Trolle.
\newblock The term structure of interbank risk.
\newblock {\em Journal of Financial Economics}, 109(3):707--733, 2013.

\bibitem[FR13]{FontanaRunggaldier13}
C.~Fontana and W.J. Runggaldier.
\newblock Diffusion-based models for financial markets without martingale
  measures.
\newblock In F.~Biagini, A.~Richter, and H.~Schlesinger, editors, {\em Risk
  Measures and Attitudes}, EAA Series, pages 45--81. Springer, London, 2013.

\bibitem[Fri75]{Friedman75}
A.~Friedman.
\newblock {\em Stochastic Differential Equations and Applications}.
\newblock Academic Press, New York, 1975.

\bibitem[GSS17]{Gallitschke17}
J.~Gallitschke, S.~Seifried, and F.T. Seifried.
\newblock Interbank interest rates: funding liquidity risk and {XIBOR} basis
  spreads.
\newblock {\em Journal of Banking and Finance}, 78:142--152, 2017.

\bibitem[GP11]{GarleanuPedersen11}
N.~Garleanu and L.H. Pedersen.
\newblock Margin-based asset pricing and deviations from the law of one price.
\newblock {\em Review of Financial Studies}, 24(6):1980--2022, 2011.

\bibitem[GM11]{GawareckiMandrekar}
L.~Gawarecki and V.~Mandrekar.
\newblock {\em Stochastic Differential Equations in Infinite Dimensions}.
\newblock Springer, Berlin - Heidelberg, 2011.

\bibitem[GR13]{GR13}
A.~Gombani and W.J. Runggaldier.
\newblock Arbitrage-free multifactor term structure models: a theory based on
  stochastic control.
\newblock {\em Mathematical Finance}, 23(4):659--686, 2013.

\bibitem[GR15]{GrbacRunggaldier}
Z.~Grbac and W.J. Runggaldier.
\newblock {\em Interest Rate Modeling: Post-Crisis Challenges and Approaches}.
\newblock Springer, Cham, 2015.

\bibitem[JS03]{JS03}
J.~Jacod and A.N. Shiryaev.
\newblock {\em Limit Theorems for Stochastic Processes}.
\newblock Springer, Berlin - Heidelberg - New York, 2nd edition, 2003.

\bibitem[KKS16]{KKS16}
Y.~Kabanov, C.~Kardaras, and S.~Song.
\newblock No arbitrage of the first kind and local martingale num{\'e}raires.
\newblock {\em Finance and Stochastics}, 20(4):1097--1108, 2016.

\bibitem[KK07]{KaratzasKardaras07}
I.~Karatzas and C.~Kardaras.
\newblock The numeraire portfolio in semimartingale financial models.
\newblock {\em Finance and Stochastics}, 11(4):447--493, 2007.

\bibitem[KK21]{KK21}
I.~Karatzas and C.~Kardaras.
\newblock {\em Portfolio Theory and Arbitrage: A Course in Mathematical
  Finance}.
\newblock American Mathematical Society, Providence (Rhode Island), 2021.

\bibitem[Kar22]{Kardaras22}
C.~Kardaras.
\newblock Stochastic integration with respect to arbitrary collections of
  continuous semimartingales and applications to mathematical finance.
\newblock {\em Annals of Applied Probability}, forthcoming, 2022.

\bibitem[Nag03]{Nagai03}
H.~Nagai.
\newblock Optimal strategies for risk-sensitive portfolio optimization problems
  for general factor models.
\newblock {\em SIAM Journal on Control and Optimization}, 41(6):1779--1800,
  2003.

\bibitem[Pas11]{Pascucci}
A.~Pascucci.
\newblock {\em {PDE} and Martingale Methods in Option Pricing}.
\newblock Springer, Milan - Dordrecht - Heidelberg - London - New York, 2011.

\bibitem[Pav22]{Pavarana22}
S.~Pavarana.
\newblock {\em A Stochastic Control Perspective of Multi-Curve Term Structures
  under the Benchmark Approach}.
\newblock Master thesis, Department of Mathematics, University of Padova,
  available at \url{https://thesis.unipd.it/handle/20.500.12608/35013}, 2022.

\bibitem[PH06]{PlatenHeath}
E.~Platen and D.~Heath.
\newblock {\em A Benchmark Approach to Quantitative Finance}.
\newblock Springer, Berlin - Heidelberg, 2006.

\bibitem[SS23]{SS22}
J.B. Skov and D.~Skovmand.
\newblock Decomposing {LIBOR} in transition: evidence from the futures markets.
\newblock {\em Quantitative Finance}, 23(6):959--978, 2023.

\end{thebibliography}

\end{document}